\documentclass[sigplan,10pt]{acmart}
\renewcommand\footnotetextcopyrightpermission[1]{} 
\pdfoutput=1

\usepackage{algorithm}
\usepackage{algpseudocode}
\usepackage{appendix}
\usepackage{xspace}
\usepackage{threeparttable}
\usepackage{color}
\usepackage{colortbl}
\usepackage {fancyhdr}
\usepackage{url}
\usepackage{subcaption}
\usepackage{paralist}
\usepackage{wrapfig}
\usepackage{multirow}
\usepackage{listings}
\usepackage{comment}
\usepackage{booktabs}
\usepackage{makecell}
\usepackage{boxedminipage}
\usepackage{amsmath}
\usepackage{graphicx}
\usepackage{subfloat} 
\usepackage{minibox}
\usepackage{pifont}
\usepackage{bbding}
\usepackage{color}
\usepackage{xcolor}
\usepackage{subfiles}
\usepackage[utf8]{inputenc}
\usepackage[english]{babel}

\usepackage{ifthen}
\newboolean{showedits}
\setboolean{showedits}{true} 

\ifthenelse{\boolean{showedits}}
{
	\newcommand{\del}[1]{\textcolor{red}{\sout{#1}}} 
}{
	\newcommand{\del}[1]{} 
	
}

\newboolean{showcomments}
\setboolean{showcomments}{true}
\newcommand{\id}[1]{$-$Id: scgPaper.tex 32478 2010-04-29 09:11:32Z oscar $-$}

\ifthenelse{\boolean{showcomments}}
{\newcommand{\nbc}[3]{
 {\colorbox{#3}{\bfseries\sffamily\scriptsize\textcolor{white}{#1}}}
 {\textcolor{#3}{\sf\small$\blacktriangleright$\textit{#2}$\blacktriangleleft$}}}
 }
{\newcommand{\nbc}[3]{}
 \renewcommand{\del}[1]{} 
 }
 \definecolor{darkyellow}{RGB}{255, 222, 9}

\definecolor{ibcolor}{rgb}{0.4,0.6,0.2}
\definecolor{ascolor}{rgb}{0,0.5,0.9}

\definecolor{bocolor}{rgb}{0.6,0.9,0.2}
\definecolor{jrcolor}{rgb}{0.5,0,0.5}
\definecolor{nrcolor}{rgb}{0.4,0.1,0.3}
\definecolor{hkcolor}{rgb}{1.0,0.5,0.3}
\definecolor{tdcolor}{rgb}{1.0,0,0}



\newcommand{\bheading}[1]{{\vspace{4pt}\noindent{\textbf{#1}}}}
\newcommand{\iheading}[1]{{\vspace{2pt}\noindent{\textit{#1}}}}
\newcolumntype{?}{!{\vrule width 1pt}}

\newcounter{note}[section]

\newcommand{\secref}[1]{\mbox{Sec.~\ref{#1}}\xspace}

\newcommand{\figref}[1]{\mbox{Fig.~\ref{#1}}}

\newcommand{\ignore}[1]{}

\newcommand{\ie}{\textit{i.e.}\xspace}
\newcommand{\eg}{\textit{e.g.}\xspace}

\newcommand{\etal}{\textit{et al.}\xspace}

\newcommand{\sysname}{\textsc{Ladon}\xspace}

\newcounter{packednmbr}

\newenvironment{packeditemize}{
\begin{list}{$\bullet$}{
\setlength{\labelwidth}{0pt}
\setlength{\itemsep}{2pt}
\setlength{\leftmargin}{\labelwidth}
\addtolength{\leftmargin}{\labelsep}
\setlength{\parindent}{0pt}
\setlength{\listparindent}{\parindent}
\setlength{\parsep}{1pt}
\setlength{\topsep}{1pt}}}{\end{list}}

\newtheorem{theorem}{Theorem}
\newtheorem{lemma}{Lemma}

\AtBeginDocument{%
  \providecommand\BibTeX{{%
    \normalfont B\kern-0.5em{\scshape i\kern-0.25em b}\kern-0.8em\TeX}}}

\begin{document}

\title{\sysname: High-Performance Multi-BFT Consensus via Dynamic Global Ordering (Extended Version)}
\thanks{This manuscript is an extended version of the paper published at EuroSys’25, DOI: 10.1145/3689031.3696102}

\author{Hanzheng Lyu}
\authornote{Both authors contributed equally to this research. This work was done in part while Hanzheng Lyu was a visiting student at SUSTech.}
\email{hzlyu@student.ubc.ca}
\affiliation{%
  \institution{University of British Columbia (Okanagan campus)}
}
\author{Shaokang Xie}
\authornotemark[1]
\email{12011206@mail.sustech.edu.cn}
\affiliation{%
  \institution{Southern University of Science and Technology}
}

\author{Jianyu Niu}
\authornote{Jianyu Niu and Yinqian Zhang are affiliated with
Research Institute of Trustworthy Autonomous Systems and Department of Computer Science and Engineering of SUSTech. Corresponding author: Jianyu Niu.}
\email{niujy@sustech.edu.cn}
\affiliation{%
  \institution{Southern University of Science and Technology}
}

\author{Chen Feng}
\email{chen.feng@ubc.ca}
\affiliation{%
  \institution{University of British Columbia (Okanagan campus)}
}

\author{Yinqian Zhang}
\authornotemark[2]
\email{yinqianz@acm.org}
\affiliation{%
  \institution{Southern University of Science and Technology}
}

\author{Ivan Beschastnikh}
\email{bestchai@cs.ubc.ca}
\affiliation{%
  \institution{University of British Columbia (Vancouver campus)}
}

\renewcommand{\shortauthors}{Lyu et al.}

\begin{abstract}
Multi-BFT consensus runs multiple leader-based consensus instances in parallel, circumventing the leader bottleneck of a single instance. 
However, it contains an Achilles’ heel: the need to globally order output blocks across instances. Deriving this global ordering is challenging because it must cope with different rates at which blocks are produced by instances. Prior Multi-BFT designs assign each block a global index before creation, leading to poor performance. 

We propose \textit{\sysname}, a high-performance Multi-BFT protocol that allows varying instance block rates. 
Our key idea is to order blocks across instances dynamically, which eliminates blocking on slow instances.
We achieve dynamic global ordering by assigning \textit{monotonic ranks} to blocks. We pipeline rank coordination with the consensus process to reduce protocol overhead and combine aggregate signatures with rank information to reduce message complexity. 
\sysname's dynamic ordering enables blocks to be globally ordered according to their generation, which respects inter-block causality.
We implemented and evaluated \sysname by integrating it with both PBFT and HotStuff protocols. 
Our evaluation shows that \sysname-PBFT (resp., \sysname-HotStuff) improves the peak throughput of the prior art by $\approx8\times$ (resp., $2\times$)  and reduces latency by $\approx62\%$ (resp., 23\%), when deployed with one straggling replica (out of $128$ replicas) in a WAN setting. 
\end{abstract}





\maketitle

\section{Introduction} \label{sec:intro} 
Byzantine Fault Tolerant (BFT) consensus is crucial to establishing a trust foundation for modern decentralized applications.
Most existing BFT consensus protocols, like PBFT~\cite{pbft1999} or HotStuff~\cite{hotstuff}, adopt a leader-based scheme~\cite{pbft1999,700BFT, HQReplica}, in which the protocol runs in views, and each view has a delegated replica, known as the leader. The leader is responsible for broadcasting proposals (\ie, a batch of client transactions) and coordinating with replicas to reach a consensus on its proposals.
However, the leader can become a significant performance bottleneck, especially at scale. The leader's workload increases linearly with the number of replicas~\cite{gai2021dissecting, MIR-BFT, avarikioti2020fnf, stathakopoulou2022state, gupta2021rcc}, making the leader the dominant factor in the system's throughput and latency. 

\begin{figure}[t]
    \centering
    \includegraphics[width=3.2in]{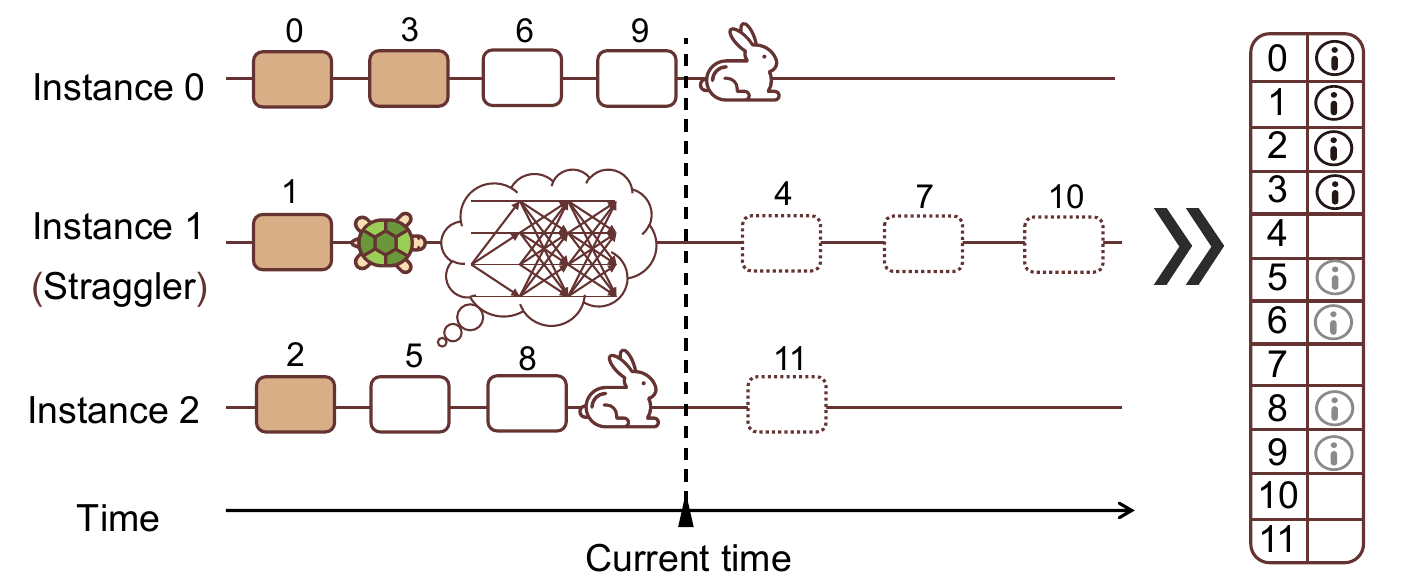}
    \caption{\textbf{An overview of Multi-BFT paradigm.} 
    The shaded (resp., white) blocks refer to the globally confirmed (resp., partially committed) blocks. The dashed blocks refer to the blocks to be produced in the future.} 
    \label{fig:MultiBFT}
\end{figure}

To address the leader bottleneck, Multi-BFT systems~\cite{MIR-BFT, avarikioti2020fnf, stathakopoulou2022state, gupta2021rcc} have emerged as a promising alternative. Multi-BFT consensus runs multiple leader-based BFT instances in \textit{parallel}, as shown in \figref{fig:MultiBFT}.
A replica may participate as a leader in one BFT instance and as a backup in others. Like a single BFT system, each BFT instance in Multi-BFT outputs a sequence of committed blocks, which will never be reverted in the partial sequence (as opposed to the global sequence introduced shortly). 
These blocks are referred to as being partially committed. Then, these blocks are ordered in a global sequence and become globally confirmed, functioning as a single instance system.
Such a scheme can balance the workloads of replicas and fully utilize their bandwidth, thereby increasing the overall system throughput.

However, the Achilles’ heel of Multi-BFT consensus lies in its global ordering. 
Existing Multi-BFT protocols~\cite{MIR-BFT, avarikioti2020fnf, stathakopoulou2022state, gupta2021rcc} 
follow a \textit{pre-determined} global ordering: a block is assigned a global index that depends solely on two numbers, its instance index and its sequence number in the instance's output. As a concrete example, consider \figref{fig:MultiBFT} with three instances outputting four blocks (produced and to be produced in the future). For example,
the three blocks from Instance 2 receive global indices of $2$, $5$, $8$, and $11$. Replicas execute partially committed blocks with an increasing global index one by one until they see a missing block. 

In a decentralized system, this simple global ordering method has performance issues.
A slow leader, often called a straggler, can slow down not just one instance, but the entire system.
For example, Instance 1 in \figref{fig:MultiBFT} has a straggling leader that only outputs one block (with a global index of $1$). This causes three ``holes'' in the global log (at positions $4$, $7$ and $10$), and prevents four blocks ($5$, $6$, $8$, and $9$) from being globally confirmed. This reduces the system's throughput and increases latency, posing a challenge for building high-performance Multi-BFT systems. (We provide further theoretical analysis and detailed experimental results to illustrate the impact of stragglers in \secref{subsec: straggler}).
This challenge is prevalent in decentralized systems where variations in replica performance and network conditions make straggling leaders commonplace.  Even worse, an adversary can violate the causality of blocks across instances, allowing it to front-run~\cite{eskandari2020sok, baum2021sok} its transactions ahead of others to gain an unfair advantage in applications like auctions and exchanges. (See more discussion in~\secref{subsec:causality}.)

In this paper, we propose \sysname\footnote{\sysname is a monster in Greek mythology, the dragon with one hundred heads that guarded the golden apples in the Garden of the Hesperides.}, a high-performance Multi-BFT consensus protocol that considers instances' varying block rates. 
Our insight is to \textit{dynamically} order partially committed 
blocks from different instances by their assigned \textit{monotonic ranks} at production. 
The rank assignment satisfies two properties: 1) \textit{agreement}: all honest replicas have the same $rank$ for a partially committed block; and, 2) \textit{monotonicity}: the ranks of subsequently generated blocks are always larger than the rank of a partially committed block to respect block causality. 
Here, the agreement property ensures that replicas that use a deterministic algorithm to order partially committed blocks by their ranks can achieve the same ordering sequence, while monotonicity 
preserves inter-block causality. 

The above dynamic global ordering decouples the dependency between the replicas' partial logs to ensure fast generation of the global log, eliminating the straggler impact. It also orders blocks by their generation sequence, preserving inter-block causality. 
For instance, consider \figref{fig:MultiBFT} with instances using the monotonic rank for subsequent blocks. The next block of Instance 1 will be assigned the rank of $10$ instead of $4$, forcing it to be globally ordered after existing partially committed blocks. The dynamic ranks enable Instance 1 to quickly synchronize with other instances, alleviating the impact of stragglers. 

Achieving monotonic ranks is challenging due to the presence of Byzantine behaviors.
Replicas need to agree on the ranks of blocks to achieve consensus. 
To maintain monotonicity, it is crucial that malicious leaders do not use stale ranks or exhaust the range of available ranks. A leader has to choose the highest rank from the ranks collected from more than two-thirds of the replicas and increase it by one.
To ensure that the leader follows these rules, each block includes a set of collected ranks and the associated proof (\ie, an aggregate signature) of the chosen ranks.  
However, this basic solution introduces latency and overhead. 
To optimize this solution, we pipeline the rank information collection with the last round of consensus and further combine aggregate signatures with rank information to reduce message complexity.  

We built end-to-end prototypes of \sysname with PBFT and HotStuff, denoted as \sysname-PBFT and \sysname-HotStuff, respectively. Furthermore, we believe \sysname can compose with any single-leader BFT protocol.
We conduct extensive experiments on AWS to evaluate and compare \sysname with existing Multi-BFT protocols, including ISS~\cite{stathakopoulou2022state}, RCC~\cite{gupta2021rcc}, Mir~\cite{MIR-BFT}, and DQBFT~\cite{dqbft}. 
We run experiments over LAN and WAN with $8-128$ replicas, distributed across $4$ regions. With one straggler in WAN, \sysname-PBFT (resp. \sysname-HotStuff) achieves $8\times$ (resp. $2\times$) higher throughput and $62\%$ (resp. $23\%$) lower latency with $128$ replicas as compared to ISS-PBFT (resp. ISS-HotStuff). Over LAN, \sysname demonstrates performance trends similar to those observed in the WAN setting.
\begin{figure}[t]
\centering
\includegraphics[width=\linewidth]{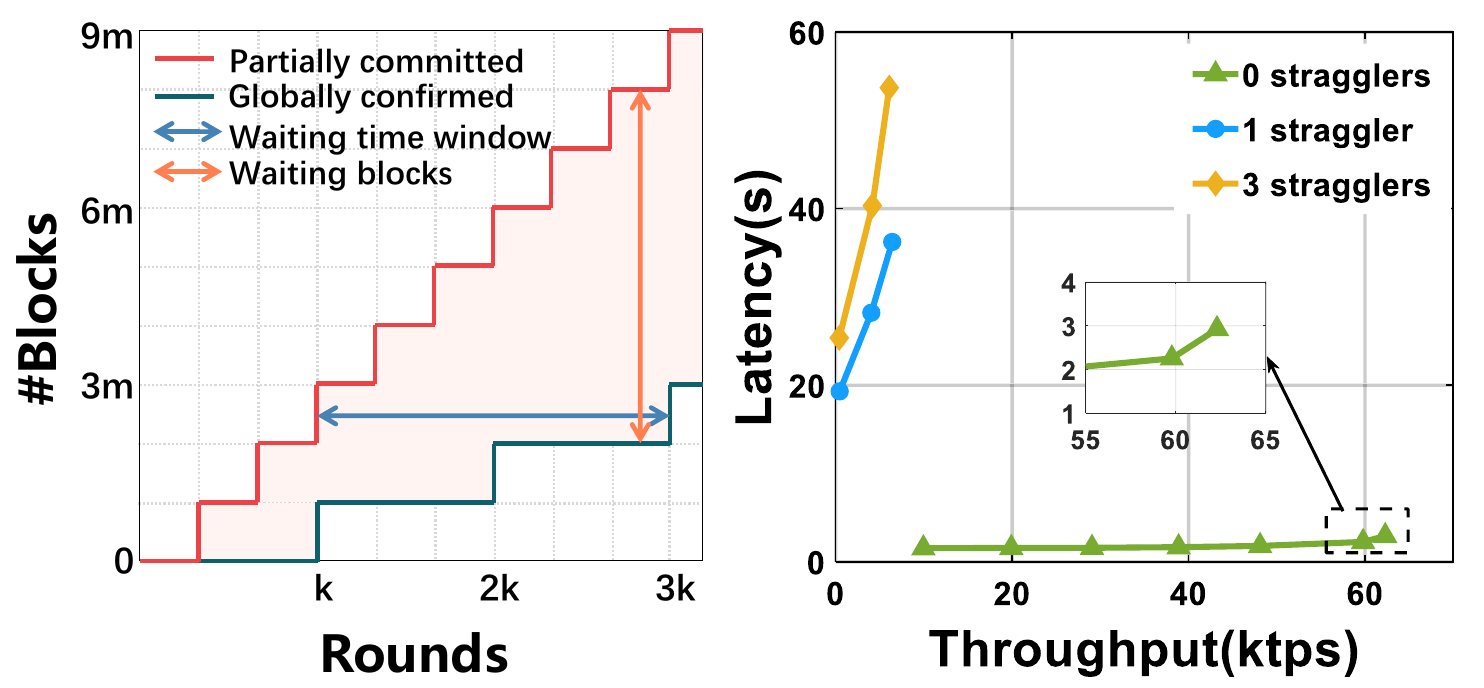}
\textbf{~~~~~(a)}~Analytical results\quad~~~~~~~~\textbf{(b)}~Experimental results
\caption{\textbf{The analytical and experimental performance of Multi-BFT consensus with/without a straggler.} The vertical line in (a) represents the queued partially committed blocks, while the horizontal line represents the delay of the global ordering.}
\label{fig:N}
\end{figure}
\section{Existing Multi-BFT Susceptibility} \label{sec:back&motivation}

\subsection{Performance Degradation} \label{subsec: straggler}
Existing Multi-BFT protocols~\cite{MIR-BFT, avarikioti2020fnf, stathakopoulou2022state, gupta2021rcc} with the pre-determined global ordering
perform well when all instances have the same block production rate. 
With $m$ instances, they can achieve $m$ times higher throughput and similar latency to single-instance systems. By adjusting $m$, the system can maximize the capacity of each replica, allowing the throughput to approach the physical limit of the underlying network.  

However, performance will significantly drop when there are straggling instances caused by faulty or limited-capacity leaders, unstable networks, or malicious behaviors. 
Consider a simple case where a slow instance with a straggling leader produces blocks every $k$ rounds while the remaining $m-1$ normal instances produce blocks every round. 
Let $R$ (resp., $R^{\prime}$) denote the number of partially committed (resp., globally confirmed) blocks per round.
We have $R = 1/k + m-1$ and $R^{\prime} =m/k$, which implies that the system throughput is about $1/k$ of the ideal scenario. 
Over time, the accumulation of $R^{\prime}-R$ blocks every round leads to a continuous delay increase in waiting for global confirmation.

\figref{fig:N}a shows the analytical results for the case above. 
First, we observe that the number of partially committed blocks that wait to be globally confirmed grows over time. 
Similarly, the delay for partially committed blocks to become globally confirmed also grows over time. 
\figref{fig:N}b plots experimental results of throughput and latency (defined in \secref{sec:expperformance}) to show the practical impact of stragglers. We run ISS~\cite{stathakopoulou2022state}, a state-of-the-art Multi-BFT protocol, in which consensus instances are instantiated with PBFT~\cite{pbft1999}. We set $m = 16$, and show results for 0, 1, and 3 stragglers in WAN. Other settings are the same as \secref{sec:expperformance}.  With 1 and 3 stragglers, the maximum throughput is reduced by $89.7\%$ and $90.2\%$ of the system's throughput with 0 stragglers, respectively. The latency with 1 and 3 stragglers increases up to $12\times$ and $18\times$  of the system's latency with 0 stragglers, respectively.

\subsection{Revisiting Straggler Mitigation} \label{subsec: revisStragg}
Most existing methods focus on detecting straggling leaders and then replacing them with normal ones. However, we show that this detect-and-replace approach cannot fix these issues. This motivated us to design a dynamic global ordering mechanism to mitigate the impact of stragglers algorithmically. 
Stathakopoulou \etal~\cite{MIR-BFT} propose to use the timeout mechanism to replace leaders of the instances that do not timely output partially committed blocks. 
Similarly, in RCC~\cite{gupta2021rcc}, a straggling leader will be removed once its instance lags behind other instances by a certain number of blocks. 
These mechanisms fall short in three ways. 
{First}, if there are multiple colluding stragglers in the system, it is difficult to detect them. 
{Second}, replicas that perform poorly due to lower capacities will be replaced, resulting in poor participation fairness for decentralized systems. 
{Third}, straggling leaders are only one cause of not timely producing partially committed blocks by instances. Network turbulence and dynamically varying replica capacities could also cause an instance to slow down for a period of time. 

Unlike these methods, DQBFT~\cite{dqbft} mitigates the impact of stragglers by adding a special instance to globally order partially committed blocks from other instances.
However, this centralized instance can itself become a performance bottleneck with a straggling leader. And, the leader in this special instance can also maliciously manipulate the global ordering of blocks. 


\section{Models and Problem Statement} \label{sec:model}
\subsection{System Model}
We consider a system with $n = 3f+1$ replicas, denoted by the set $\mathcal{N}$. 
A subset of at most $f$ replicas is \textit{Byzantine}, denoted as $\mathcal{F}$.
Byzantine replicas can behave arbitrarily. 
The remaining replicas in $\mathcal{N} \setminus \mathcal{F}$ are honest and strictly follow the protocol. 
All the Byzantine replicas are assumed to be controlled by a single adversary, which is computationally bounded. Thus, the adversary cannot break the cryptographic primitives to forge honest replicas' messages (except with negligible probability). 
There is a public-key infrastructure (PKI): each replica has a pair of keys for signing messages.

We assume honest replicas are fully and reliably connected: every pair of honest replicas is connected with an authenticated and reliable communication link.
We adopt the partial synchrony model of Dwork \etal \cite{dwork1988consensus}, which is widely used in BFT consensus~\cite{pbft1999, hotstuff}.
In the model, there is a known bound $\Delta$ and an unknown Global Stabilization Time (\textsf{GST}), such that after \textsf{GST}, all message transmissions between two honest replicas arrive within a bound $\Delta$. 
Hence, the system is running in \textit{synchronous} mode after \textsf{GST}.

Like classical BFT consensus~\cite{pbft1999, hotstuff}, we assume that Byzantine replicas aim to destroy the \textit{safety} and \textit{liveness} properties by deviating from the protocol. They may strategically delay their operation (e.g., without triggering timeout mechanisms), appearing as stragglers to either compromise performance or violate block causality (\secref{subsec:causality}). It is worth noting that in decentralized applications, due to replicas' heterogeneous capacities, honest replicas may also behave as stragglers.  

\subsection{Preliminaries}\label{subsec:preliminaries}
We now introduce some building blocks for our work.

\bheading{Sequenced Broadcast (SB).} We follow ISS~\cite{stathakopoulou2022state} to use Sequenced Broadcast (SB) for consensus instances. SB is a variant of Byzantine total order broadcast with explicit round numbers and an explicit set of allowed messages. In particular, given a set of messages $M$ and a set of round numbers $R$, only one sender $p$ (\ie, the leader) can \textit{broadcast} a message $(msg, r)$, where $(msg, r) \in M \times R$. Honest replicas can \textit{deliver} a message $msg$ with round number $r$. 
If an honest replica suspects that $p$ is quiet, all correct nodes
can deliver a special $nil$ value $msg = \perp$ $\notin M$. Otherwise, honest replicas can deliver non-nil messages $m \neq \perp$.
There is a failure detector $D$ to detect a \textit{quiet} sender. SB is implementable with consensus, Byzantine reliable broadcast (BRB), and a Byzantine fault detector~\cite{stathakopoulou2022state}. 
An instance of SB ($p, R, M, D$) has the following properties:
\begin{packeditemize}
    \item \textbf{SB-Integrity}: If an honest replica delivers $(msg, r)$ with $msg \neq \perp$ and $p$ is honest, then $p$ broadcast $(msg, r)$.

    \item \textbf{SB-Agreement}: If two honest replicas deliver $(msg, r)$ and $(msg^{\prime}, r)$, then $msg = msg^{\prime}$. 

    \item \textbf{SB-Termination}: If $p$ is honest, then $p$ eventually delivers a message for every round number in $R$, \ie, $\forall r \in R:$ $\exists msg \in M \cup \{\perp\}$ such that $p$ delivers $(msg, r)$.

\end{packeditemize}

\bheading{Blocks.} A block $B$ is a tuple $(txs,index,round,rank)$, where $txs$ denotes a batch of 
clients' transactions, $index$ denotes the index of the consensus instance, $round$ denotes the proposed round, and $rank$ denotes the assigned monotonic rank.
We use $B.x$ to denote the associated parameter $x$ of block $B$. For example, $B.txs$ is the set of included transactions. If two blocks have the same $index$, we call them intra-instance blocks; otherwise, we refer to them as inter-instance blocks.
It's important to note that when a block is globally confirmed (as introduced shortly), replicas can compute a unique global ordering index $sn$ for it. 
In other words, $sn$ is not a predetermined field of the block. In the following, we still use $B.sn$ for clarity. 

\bheading{Aggregated signature scheme.} The aggregated signature is a variant of the digital signature that supports aggregation~\cite{aggregateSig}. That is, given a set of users $R$, each with a signature $\sigma_r$ on the message $m_r$, the generator of the aggregated signature can aggregate these signatures into a unique short signature: ${\textsc{agg}} (\{\sigma_r\}_{r \in R}) \rightarrow \sigma$. Given an aggregation signature, the identity of the aggregation signer $r$ and the original message $m_r$ of the signature can be extracted. The verifier can also verify that $r$ signed message $m_r$ by using the verify function: ${\textsc{verifyAgg}} ((pk_r,m_r)_{r \in R},\sigma) \rightarrow 0/1$.

\subsection{Problem Formulation}\label{subsec:formulation}
We consider a Multi-BFT system consisting of $m$ BFT instances indexed from 0 to $m-1$. Thus, the $i$th ($1\leq i \leq m$) instance has an index of $i-1$.
The system can be divided into two layers, a \textit{partial ordering layer} $\mathcal{P}$ and a \textit{global ordering layer} $\mathcal{G}$.

\bheading{Partial ordering layer.} Clients create and send their transactions to replicas for processing, which constitute the input of the partial ordering layer $\mathcal{P}_{in}$. We assume there is a mechanism (\eg, rotating bucket~\cite{MIR-BFT}), which assigns client transactions to different instances to avoid transaction redundancy. Each instance runs an SB protocol, and in each round, a leader packs client transactions into blocks, proposes (\ie, \textit{broadcast} in SB) the blocks, and coordinates all replicas to continuously agree on the blocks. We denote the block produced by instance $i$ in round $j$ as $B_j^i$. 
The output of the partial ordering layer is a collection of $m$ sequences of \textit{partially committed} (\ie, \textit{delivered} in SB) blocks produced by all the instances, where the $i$th sequence from $i$th instance is denoted by $\big < B_1^i, B_2^i, $ $..., B_{k_i}^i \big>$, and $k_i$ is the number of partially committed blocks in the $i$th instance till now. We denote the entire collection by $\mathcal{P}_{out} = \big < B_1^i, B_2^i, $ $..., B_{k_i}^i \big>_{i={0}}^{m-1}$.

\bheading{Global ordering layer.} A Multi-BFT system should perform as a single BFT system, and so blocks output by the partial ordering layer across $m$ instances should be ordered into a global sequence. 
Thus, the input of the global ordering layer is the output of the partial ordering layer, \ie, $\mathcal{G}_{in} = \mathcal{P}_{out} = \big < B_1^i, B_2^i, ..., B_{k_i}^i \big>_{i={0}}^{m-1}$.
Following certain ordering rules (which vary according to different designs), these input blocks are ordered into a sequence, and the associated index is denoted as $sn$. These output blocks are globally committed and executed, \ie, \textit{globally confirmed}, denoted by the set $\mathcal{G}_{out}$. Blocks in $\mathcal{G}_{out}$ satisfy the following properties:
\begin{packeditemize}
\item \textbf{$\mathcal{G}$-Agreement}: If two honest replicas globally confirm $B.sn = B'.sn$, then $B=B'$. 

\item \textbf{$\mathcal{G}$-Totality}: If an honest replica globally confirms a block $B$, then all honest replicas eventually globally confirm the block $B$.

\item \textbf{$\mathcal{G}$-Liveness}: If a correct client broadcasts a transaction $tx$, an honest replica eventually globally confirms a block $B$ that includes $tx$. 

\end{packeditemize}

\section{Dynamic Global Ordering}\label{sec: keycompnent}

\subsection{Monotonic Rank}\label{subsec:rank}
We first present the required properties of monotonic ranks and then discuss how to realize them.

\bheading{Properties.} The global ordering layer assigns partially committed blocks monotonic ranks (short for rank in the following discussion). These ranks will determine the output block sequence of the system. 
\textbf{M}onotonic \textbf{r}anks have two key properties: 
\begin{packeditemize}
    \item \textbf{MR-Agreement:} All honest replicas have the same rank for a partially committed block.

    \item \textbf{MR-Monotonicity:} 
     If a block $B'$ is generated after an intra-instance (or a partially committed inter-instance) block $B$, then the rank of $B'$ is larger than the rank of $B$.
\end{packeditemize}

These two properties collectively ensure that given a set of partially committed blocks with ranks, honest replicas can independently execute an ordering algorithm to produce a consistent sequence of globally confirmed blocks without any additional communication. Specifically, blocks are ordered by increasing rank, with a tie-breaking of instance indexes.
MR-Monotonicity guarantees block causality such that if a block is partially committed before another block is generated, then the partially committed block is globally confirmed before the latter. While a trusted global time could further strengthen block causality~\cite{kelkar2020order}, the current system still provides a robust level of causality, even without such global time coordination.




\bheading{The realization.} We first show how to achieve the MR-Agreement property for blocks' ranks without running additional consensus protocols. In particular, a block can be assigned a rank either when it is proposed or after it is output by a consensus instance (\ie, running SB). We observe that for the former approach, no additional procedure is required to achieve the property since the rank is piggybacked with the block that has to go through the consensus process. 
By contrast, the latter approach requires an extra consensus process. Thus, we use the former approach for efficiency. 

Second, to achieve MR-Monotonicity, a leader first collects the highest ranks from at least $2f+1$ replicas. It then increments the highest rank among these by one and assigns this as the rank for its proposed block. 
However, there exists a potential vulnerability: a malicious leader might attempt to disrupt monotonicity by introducing stale ranks.
To counteract this, each collected rank is sourced from blocks that have received sufficient votes and are accompanied by cryptographic certificates, which validate their authenticity (see \secref{subsec:pbft}). These collected ranks and certificates are then integrated into the proposed block for validation. Consequently, even if a Byzantine leader attempts to manipulate the ranks, the authenticity checks constrain it. We further analyze the impact of possible Byzantine behaviors on system performance in~\secref{subsec:byzantinestraggler}.

\bheading{Overhead analysis.} The above approach has two overheads: the rank collection process (one round of communication) and the larger block size. 
To mitigate the former, we integrate the rank collection into the consensus phases of the prior block. This eliminates extra communication, significantly reducing rank collection latency. As for the latter, we observe that the increased block size is negligible given the block payload size. For example, the additional rank information and certificates included in blocks comprise less than 1\% of the total block size (\ie, 2MB) with 100 replicas. We also use a custom aggregate signature mechanism to further reduce rank information included in blocks.
Specifically, a block only needs to include one aggregated signature as a certificate of the collected ranks.



\subsection{Global Ordering Algorithm}\label{subsec:gor}
Algorithm~\ref{main} shows the global ordering process running at a replica.
The algorithm takes the set $\mathcal{G}_{in}$ of partially committed blocks as its input (which is the output from the partial ordering layer) and outputs 
the set $\mathcal{G}_{out}$ of globally confirmed blocks (\secref{subsec:formulation}). 

Algorithm~\ref{main} is based on two basic ideas. First, partially committed blocks can be globally ordered by increasing ranks and a tie-breaking to favor block output from consensus instances with smaller indices.
For example, given two blocks $B$ and $B^{\prime}$, block $B$ will be globally ordered before $B^{\prime}$, when $B.rank < B^{\prime}.rank$ or $B.rank = B^{\prime}.rank \wedge B.index < B^{\prime}.index$. 
For convenience, we use $B \prec B^{\prime}$ to denote that block $B$ has a lower global ordering index than block $B^{\prime}$.
Second, a partially committed block can be globally confirmed if its global index is lower than a certain threshold (which will be defined later).

\begin{algorithm}[t]
\caption{The Global Ordering Algorithm 
}
\label{main}
\begin{algorithmic}[1]  
\State \textbf{upon} $\mathcal{G}_{in}$ is updated 
\State \hspace{1.0em} $\mathcal{S'} \leftarrow$ \textsc{getLastBlock}($\mathcal{G}_{in}$)
\State \hspace{1.0em} $B^{*} \gets \textsc{findLowestBlock}(\mathcal{S'})$ 
\State \hspace{1.0em} $bar = (B^{*}.rank + 1,B^{*}.index)$ \textcolor{purple}{//compute the bar}
\State \hspace{1.0em} $\mathcal{S} \gets \mathcal{G}_{in} \setminus \mathcal{G}_{out}$ 
\State \hspace{1.0em} $B_{can} = \textsc{findLowestBlock}(\mathcal{S})$ \textcolor{purple}{//find candidate block}
\State \hspace{1.0em} \textbf{while} $B_{can} \prec bar$ \textcolor{purple}{//$B_{can}$ has a lower index than $bar$}
\State \hspace{2.0em} $\mathcal{G}_{out} \leftarrow \mathcal{G}_{out} \cup  B_{can}$ \textcolor{purple}{//globally confirm $B_{can}$}
\State \hspace{2.0em} $\mathcal{S} \gets \mathcal{S} \setminus B_{can}$ \textcolor{purple}{//update $\mathcal{S}$}
\State \hspace{2.0em} $B_{can} = \textsc{findLowestBlock}(\mathcal{S})$ \textcolor{purple}{//find next $B_{can}$}
\State \hspace{1.0em} \textbf{end while}


\Statex
\Statex \textcolor{purple}{//Return block with the lowest ordering index}
\State \textbf{function} {\textsc{findLowestBlock}}($\mathcal{V}$)
\State  \hspace{1.0em} $B^{*} \gets $ first block in $\mathcal{V}$
\State \hspace{1.0em} \textbf{for} each $B\in \mathcal{V}$ \textbf{do}  
\State  \hspace{2.0em} \textbf{if} $B \prec B^{*}$
\State  \hspace{3.0em} $B^{*} \gets B$
\State \hspace{2.0em} \textbf{end if}
\State \hspace{1.0em} \textbf{end for}
\State \hspace{1.0em} \textbf{return} $B^{*}$
\end{algorithmic}
\end{algorithm}

Now we describe the algorithm. When $\mathcal{G}_{in}$ is updated with new partially committed blocks, the replica runs the global ordering algorithm to decide globally confirmed blocks. The key step is to compute a threshold called the confirmation bar (short for \textit{bar}), by which blocks with lower global ordering indices can be globally confirmed.
To compute $bar$, the replica first fetches the last partially confirmed block from each instance, denoted by the set $\mathcal{S}^{\prime}$ (Line 2). 
Here, a block is partially confirmed only if all previous blocks in the same instance become partially committed. 
It then finds the block $B^{*}$ that has the lowest ordering index among the blocks in $\mathcal{S}^{\prime}$, \ie, $\forall B^{\prime} \in \mathcal{S}^{\prime} \mbox{ with } B^{\prime} \ne B^{*}: B^{*} \prec B^{\prime}$ (Line 3). 
Thereafter, \textit{bar} can be computed as a tuple of $(rank, index)$ (Line 4):  
\[bar := (rank, index) = (B^{*}.rank + 1, B^{*}.index).\]
The threshold $bar$ represents the lowest global ordering index that can be owned by subsequently generated blocks. The $bar$ is initialized with $(0,0)$. 

With $bar$ defined, the replica repetitively checks unconfirmed blocks in $\mathcal{G}_{in}$ and decides which blocks to confirm.
Specifically, let $\mathcal{S} = \mathcal{G}_{in} \setminus \mathcal{G}_{out}$ be the set of unconfirmed blocks in $\mathcal{G}_{in}$ (Line 5).
The replica finds the block  $B_{can} \in \mathcal{S}$ that has the lowest ordering index, which is referred to as the \textit{candidate} block (Line 6).
If $B_{can}$ has a lower ordering index than $bar$, $B_{can}$ can be globally confirmed because all future blocks will have higher indices than $B_{can}$. 
In particular, $B_{can}$ will be added to the set $\mathcal{G}_{out}$ and removed from the set $\mathcal{S}$ (Lines 8-9). 
The process repeats until no such $B_{can}$ can be found (Line 10).

\figref{fig:order} provides a concrete example of the global ordering process. 
Suppose at time $t_1$, a new partially committed block $B_2^2$ is added to Instance $2$, a replica has $\mathcal{G}_{in} = \big< B_1^{0}, B_2^{0}, B_3^{0}, B_1^1, B_2^1,$ $ B_1^2, B_2^2\big>$ and $\mathcal{G}_{out} = \big< B_1^{0}, B_2^{0}, B_1^1,$ $ B_1^2\big>$. 
According to the above algorithm, we have $\mathcal{S}^{\prime} = \big< B_3^{0}, B_2^{1}, B_2^{2}\big>$, $B^{*} = B_{2}^1$, $bar = (3,1)$, and $\mathcal{S} = \big< B_3^{0}, B_2^{1}, B_{2}^{2}\big>$.
The first candidate block in $\mathcal{S}$ is $B_2^{1}$, which has a lower rank than $bar$ and so will be globally confirmed. 
Then, $B_2^{1}$ is removed from $\mathcal{S}$, and $\mathcal{S} = \big< B_3^{0}, B_{2}^{2}\big>$, the candidate block is $B_3^{0}$, which can be globally confirmed because it has the same rank but a smaller index than $bar$ and so it will be globally confirmed. 
Finally, the set $\mathcal{S}$ contains only block ${B_2^2}$ with a higher rank than $bar$. Thus, ${B_2^2}$ will not be globally confirmed, and the search for globally confirmed blocks ends. 

\begin{figure}
    \centering
    \includegraphics[width=3.4in]{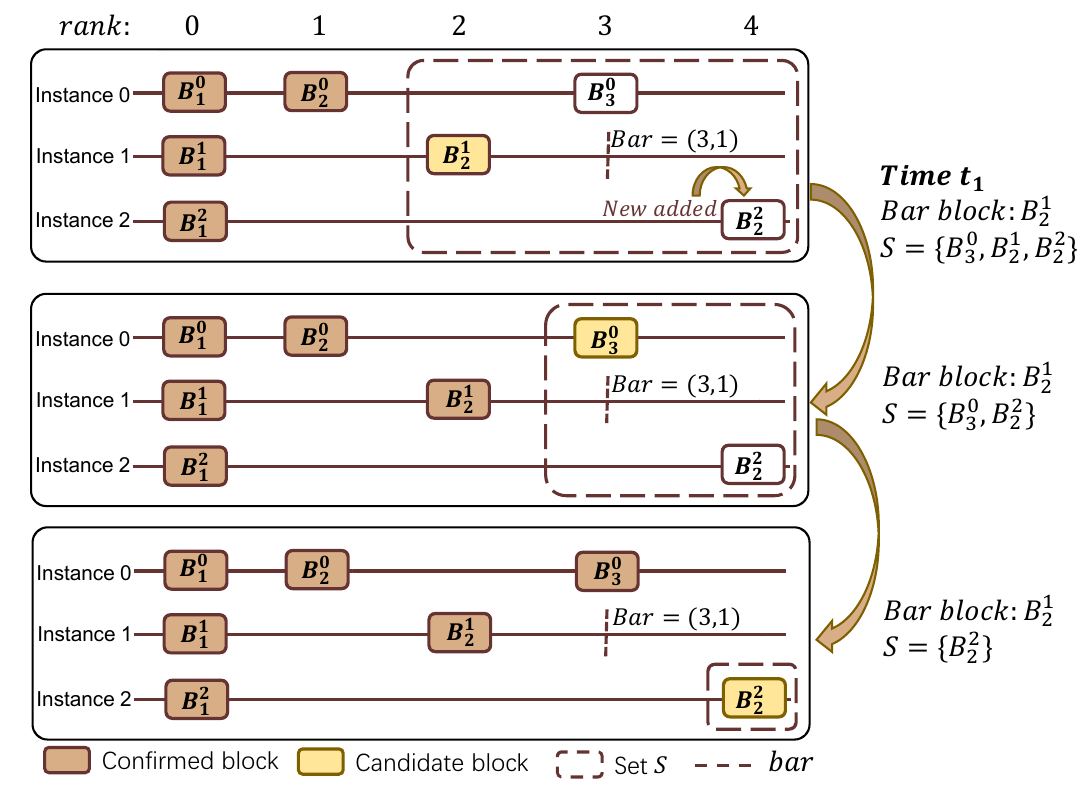}
    \caption{\textbf{An illustration of the dynamic global ordering process.} At time $t_1$, a new block $B_2^2$ is partially committed, which makes blocks $B_2^1$ and $B_3^0$ globally confirmed.
     }
    \label{fig:order}
\end{figure}




The dynamic global ordering effectively mitigates the impact of stragglers on performance. Each instance works akin to a separate relay track in a race. Instead of producing blocks in a rigid sequential order, the assigned ranks can dynamically adjust the position of each block produced by an instance. This is like allowing slower runners on a relay track to leap ahead in the race to keep up with faster runners on other tracks, improving efficiency and reducing the latency of global ordering.

\subsection{Causality Enhancement}\label{subsec:causality}
The above dynamic global ordering also respects inter-block causality, which is a highly desirable property for decentralized applications such as auctions and exchanges. 
It ensures that no one can front-run a partially committed transaction. In sharp contrast, this property is missing in previous Multi-BFT protocols (with a pre-determined ordering) \cite{MIR-BFT, avarikioti2020fnf, stathakopoulou2022state, gupta2021rcc}.
To better understand the issues in existing Multi-BFT protocols, refer to \figref{fig:MultiBFT}, where block $4$ is proposed after blocks $5$, $6$, $8$, and $9$ but is globally ordered and executed before them. Such violations may lead to various attacks including front-running attacks~\cite{eskandari2020sok, baum2021sok}, undercutting attacks~\cite{gong2022towards}, and incentive-based attacks~\cite{instability,niu2019selfish, eyal2014majority}. 
For example, 
consider a front-running attack~\cite{eskandari2020sok, baum2021sok} of cryptocurrency exchange, in which an attacker sees a large buy order $tx_v$ in block $5$, shown in \figref{fig:MultiBFT}. Then, the attacker creates a similar buy order $tx_m$ in block $4$. Since block 4 is placed ahead of block 5 in the global ordering, $tx_m$ is processed before $tx_v$. As a result, the attacker can buy the cryptocurrency at a lower price, and later sell it back at a higher price to $tx_v$, profiting at the expense of the original buyer.

Nevertheless, there is still a gap between the above property and an ideal property that no one can front-run a transaction, which is referred to as client-side causality in prior work~\cite{duan2018beat,duan2017secure, stathakopoulou2021adding}. 
Achieving this strong causality usually requires complicated fair-ordering mechanisms~\cite{clement2009making}, costly cryptographic techniques (\eg, commitment~\cite{duan2017secure} and threshold cryptography~\cite{duan2018beat}) or Trusted hardware (\eg, TEEs~\cite{stathakopoulou2021adding}). We leave it as future work to use these techniques in \sysname.

\subsection{Analysis of Byzantine leaders' Impact}\label{subsec:byzantinestraggler}
While the dynamic global ordering effectively mitigates stragglers' impact on performance as well as enhancing the inter-block causality, challenges arise with Byzantine leaders, who may intentionally delay block proposals or strategically minimize ranks of proposed blocks.

\bheading{Strategic delay of block proposals.} A Byzantine leader can strategically delay the proposal of a new block, impacting system latency. Additionally, such delays can be exploited to gain a front-running advantage or more fees from proposing transactions in some applications like blockchains. This challenge is inherent to BFT systems (including single instance and Multi-BFT systems), as delayed block proposals are a common Byzantine behavior that is hard to differentiate from honest actions in consensus protocols. 
The impact of this strategy is constrained by a timeout mechanism. In each round, the leader is given a limited window of time to propose a block. If the leader fails to do so within this period, a timeout triggers, preventing excessive delays and allowing the system to move forward, typically by initiating a view change to replace the leader. 

\bheading{Minimizing rank for proposed block.} 
Before proposing a block, a leader first collects at least $2f+1$ highest certified ranks from replicas and then increases the highest of these by one to determine the rank for its new block. However, a Byzantine leader could collect more than $2f+1$ ranks before proposing a new block, subsequently discarding several high ranks and selecting only the lowest $2f+1$ ranks (see detailed example in Appendix \ref{appen:example}).

To analyze this, consider a Byzantine leader collects at least $2f + 1$ ranks and only selects the lowest $2f+1$ ranks. Let $f'$ be the actual number of Byzantine replicas (where $f' \leq f$). Then, among the lowest $2f+1$ ranks, there are at least $2f+1-f'$ ranks from honest replicas. So the selected highest rank from these $2f+1$ ranks is greater or equal to the highest one from the $2f+1-f'$ ranks from honest replicas, which is at least the median of all the $n-f'$ ranks from all honest replicas.
This implies that the selected highest rank is at least the median of the ranks from honest replicas.

In \sysname-PBFT (see detailed description in \secref{subsec:pbft}), when a block is partially committed, at least $2f+1$ replicas send commit messages. These replicas have all received $2f+1$ prepare messages, which serve as a quorum certificate for the block's rank. Hence, at least $2f+1$ replicas (including $f+1$ honest replicas) have this certified rank. Consequently, the median of the certified ranks among all honest replicas is greater than or equal to the rank of all committed blocks. Therefore, even with rank manipulation, the leader's proposed rank will not be lower than the ranks of all partially committed blocks. 

\bheading{Conclusion.} Both Byzantine strategies will have a limited impact on system performance and transaction causality, as their effects are mitigated by built-in timeout mechanisms and rank certification.

\section{\sysname Design} \label{sec:design}
We overview \sysname in \secref{subsec:overview} and then introduce the detailed protocol with PBFT consensus instances in \secref{subsec:protocol}, and then provide refinements to reduce  message complexity in \secref{subsec: optimization}.
We sketch the correctness analysis of \sysname in \secref{subsec:correct}.
Additionally, we provide an analysis of message complexity, examples of protocol behavior, and the design for composing \sysname with HotStuff; these are detailed in Appendices \ref{appen: optimization}, \ref{appen:example}, and \ref{appen:hotstuff} due to space constraints.

\begin{figure}[t]
    \centering
    \includegraphics[width=3.5in]{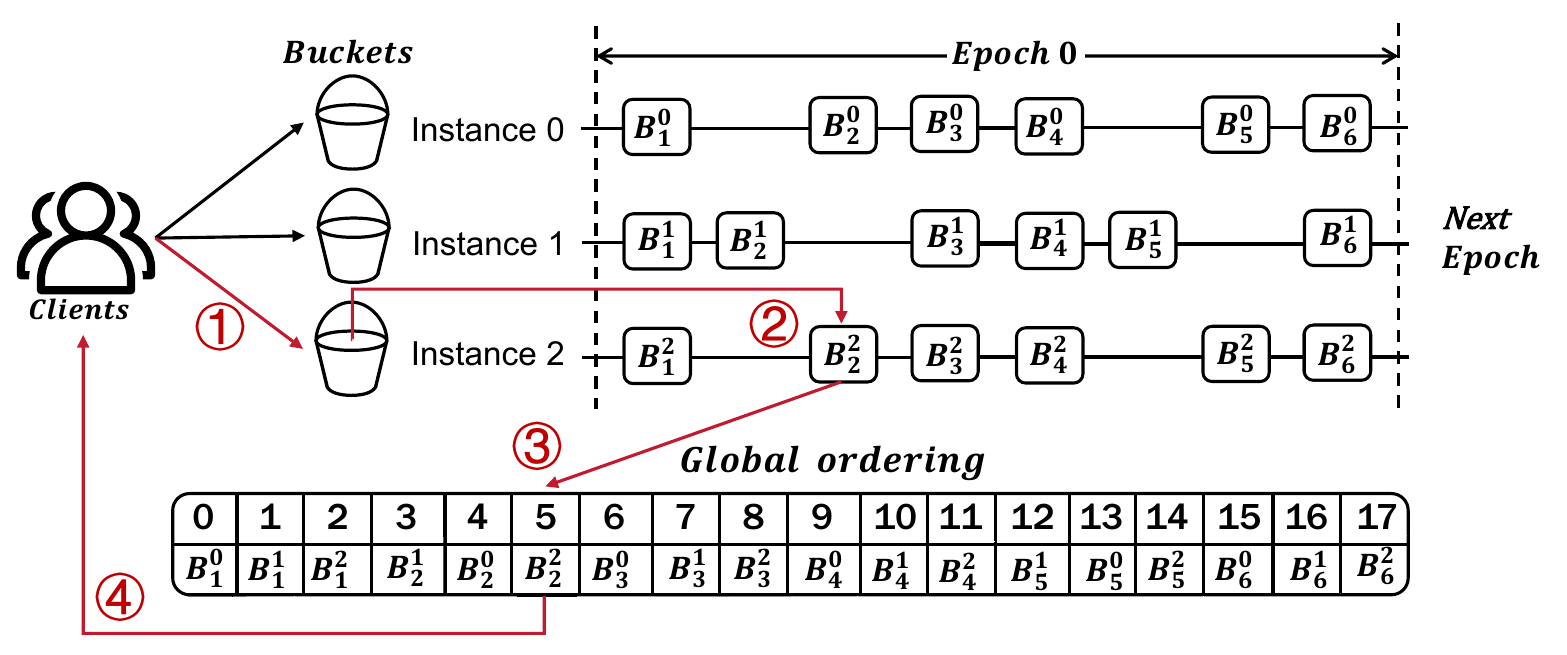}
    \caption{\textbf{An overview of \sysname's design.} Client transactions are enqueued into rotating buckets and are then consumed by concurrent BFT protocol instances. Each instance packs transactions into blocks that are eventually totally ordered. Instances execute in epochs. 
    } 
    \label{fig:overview}
\end{figure}

\subsection{Overview}\label{subsec:overview}
\figref{fig:overview} provides an overview of the core four components in \sysname: rotating buckets, epoch pacemaker, consensus instance, and global ordering. Among them, rotating buckets and the epoch pacemaker bear resemblance to existing paradigms~\cite{stathakopoulou2022state}, while instance consensus and global ordering are tailored to realize the dynamic global ordering.
%

\bheading{Rotating buckets.} 
\sysname adopts rotating buckets from ISS~\cite{stathakopoulou2022state}. 
These are used to prevent multiple leaders from simultaneously including the same transaction in a block. Client transactions are divided into disjoint buckets, and these are assigned round-robin to consensus instances when an epoch changes. Bucket rotation mitigates censoring attacks, in which a malicious leader refuses to include transactions from certain clients.

\bheading{Epoch pacemaker.}  The epoch pacemaker ensures \sysname proceeds in epochs. At the beginning of an epoch, \sysname has to configure the number of consensus instances, and the associated leader for each instance, and initialize systems parameters. At the end of an epoch, \sysname creates checkpoints of the current epoch across all instances, partially committing the block with the maximum rank of the current epoch. The liveness property inherent to each instance ensures that eventually, every instance will partially commit the block with the maximum rank for the given epoch. Once this condition is satisfied, \sysname transitions to the next epoch. We detail the Epoch pacemaker in~\secref{subsubsec:pacemaker}.

\bheading{Consensus instances.} In each epoch, multiple consensus instances run in parallel to handle transactions from rotating buckets. \sysname uses off-the-shelf BFT protocols such as PBFT and HotStuff. 
Each consensus instance contains (1) a mechanism for normal-case operation, and (2) a view-change mechanism. 
\secref{subsec:pbft} further details these two mechanisms. 

\bheading{Global ordering.} The blocks produced by consensus instances are globally ordered by the global ordering algorithm, as detailed in~\secref{subsec:gor}. Upon global confirmation of a block, the transactions are sequentially executed, with the results subsequently relayed back to the respective clients. 

We now review the flow of client transactions in \sysname. These four steps match the numbered red arrows in \figref{fig:overview}.

\vspace{1mm} \noindent \ding{172} A client creates a transaction $tx$, and sends it to some relay replicas. The transaction $tx$ is assigned into one bucket, and forwarded to the associated leader for the current epoch. 

\vspace{1mm} \noindent \ding{173} When the leader receives the transactions, it will first pack the transactions in a block. Then, it runs consensus instances with other replicas to partially commit the block. 

\vspace{1mm} \noindent \ding{174} Replicas run the global ordering algorithm to globally confirm output blocks from instances.

\vspace{1mm} \noindent \ding{175}  Once the block is globally confirmed, the replica sends a reply to the client. Upon receiving more than $f+1$ identical replies, the client acknowledges the response as accurate.

\subsection{Protocol Description}\label{subsec:protocol}

\subsubsection{Epoch Pacemaker} \label{subsubsec:pacemaker}
\sysname proceeds in epochs.
We start with epoch 0, and an empty undelivered block set.  All buckets are initially empty.

\bheading{Epoch initialization.}
At the start of epoch $e$, \sysname will do the following: (1) calculate the range of $rank$s and instance index, 
(2) calculate the set of replicas that will act as leaders in epoch $e$ based on the leader selection policy, 
(3) create a new consensus instance for each leader, (4) assign buckets and indices to the created instances,
(5) start the instances (see details in~\secref{subsec:pbft}).

Here, we omit the details of the leader selection policy, bucket assignment policy, and index assignment policy, because they are the same as the policies in ISS~\cite{stathakopoulou2022state}.

The range of $rank$s for epoch $e$ is denoted as $[minRank(e)$, $maxRank(e)]$. Given an epoch $e$, the length $l(e)$ is a customizable parameter. It characterizes the number of $rank$s in an epoch, such that $l(e) = maxRank(e) - minRank(e) + 1$.  For epoch $e = 0$, $minRank(0) = 0$, and $maxRank(0)=l(0)-1$.  For epoch $e \neq 0$, $minRank(e) = maxRank(e-1)+1$.


\bheading{Epoch advancement.}
\sysname advances from epoch $e$ to epoch $e+1$ when all the instances reach  $maxRank(e)$, \ie, the blocks with $maxRank(e)$ in all instances have been partially committed. Only then does the replica start processing messages related to epoch $e+1$. 
To prevent transaction duplication across epochs, \sysname requires replicas to globally confirm all blocks in epoch $e$ before proposing blocks for epoch $e+1$.
To this end, \sysname adopts a \textsc{checkpoint} mechanism, in which replicas broadcast a checkpoint message in the current epoch before moving to the next epoch. 
Upon receiving a quorum of $2f + 1$ valid checkpoint messages, a replica creates a stable checkpoint for the current epoch, which is an aggregation of the checkpoint messages. When a replica starts receiving messages for a future epoch $e+1$, it fetches the missing log entries of epoch $e$ along with their corresponding stable checkpoint, which prove the integrity of the data.

\begin{algorithm}[t]
\caption{The \sysname-PBFT Algorithm for Instance $i$ at View $v$, Round $n$ and Epoch $e$}
\label{alg:pbftnew}
  \begin{algorithmic}[1]
  \renewcommand{\algorithmicrequire}{\hspace{1.8em}\textcolor{purple}{\Comment{{\textsc{pre-prepare}} phase (only for leader)}}}
\Require
  \State \hspace{0.0em} \textbf{upon} receive $2f+1$ $rankmsg$ \textbf{do} 
  \State \hspace{1.0em} $ rankSet \leftarrow 2f+1$ $rankmsg$
  \State \hspace{1.0em} $ rank_m, QC \leftarrow$\textsc{getRank}$(rankSet)$
  \State \hspace{1.0em} $txs \leftarrow$ \textsc{cutBatch}$(ins.bucketSet)$
  \State \hspace{1.0em} $d \leftarrow hash(txs)$
  \State \hspace{1.0em} $rank \leftarrow min\{rank_m +1, maxRank(e)\}$ 
  \State \hspace{1.0em} $ppremsg \leftarrow \langle${\textsc{pre-prepare}}$, v, n, d, i, rank\rangle_\sigma$
  \State \hspace{1.0em} multicast  $\langle ppremsg, txs, QC, rankSet\rangle$
  \State \hspace{1.0em} \textbf{if} $rank = maxRank(e)$
  \State \hspace{2.0em} stop propose
  \State \hspace{1.0em} \textbf{end if}
  \State
  \renewcommand{\algorithmicrequire}{\hspace{1.8em}\textcolor{purple}{\Comment{{\textsc{prepare}} phase}}}
\Require
  \State \textbf{upon} receive $ppremsg$ \textbf{do} 
  \State \hspace{1.0em} \textbf{if} {\textsc{verify}}$(ppremsg)$
  \State \hspace{2.0em} $premsg \leftarrow \langle${\textsc{prepare}}$, v, n, d, i, rank\rangle_\sigma$  
  \State \hspace{2.0em} multicast $premsg$ 
  \State \hspace{1.0em} \textbf{end if}
  
  \State
  \renewcommand{\algorithmicrequire}{\hspace{1.8em}\textcolor{purple}{\Comment{{\textsc{commit}} phase}}}
\Require
  \State \textbf{upon} receive $2f+1$ $premsg$ \textbf{do}
  \State \hspace{1.0em} \textbf{if} {\textsc{verify}}$(premsg)$
  \State \hspace{2.0em} $commsg \leftarrow \langle${\textsc{commit}} $, v, n, d, i, rank\rangle_\sigma$  
  \State \hspace{2.0em} multicast $commsg$
  \State \hspace{2.0em} \textbf{if} $commsg.rank > curRank.rank$  
  \State \hspace{3.0em} $curRank.rank \leftarrow commsg.rank$
  \State \hspace{3.0em} $curRank.QC \leftarrow$ {\textsc{agg}}$(premsg)$
  \State \hspace{2.0em} \textbf{end if}
  \State \hspace{2.0em} $rankmsg \leftarrow \langle${\textsc{rank}}$, v, n, \bot, i, curRank.rank\rangle_\sigma$  
  \State \hspace{2.0em} send $\langle rankmsg, curRank.QC \rangle$ to leader 
  \State \hspace{1.0em} \textbf{end if}
  \State
  \renewcommand{\algorithmicrequire}{\hspace{1.8em}\textcolor{purple}{\Comment{Finally}}}
\Require
  \State \textbf{upon} receive $2f+1$ $commsg$ \textbf{do}
  \State \hspace{1.0em} \textbf{if} {\textsc{verify}}$(commsg)$  
  \State \hspace{2.0em} $B \leftarrow \langle txs, i, n, rank\rangle$  
  \State \hspace{2.0em} $\mathcal{G}_{in} \leftarrow \mathcal{G}_{in} \cup  B$ \textcolor{purple}{{//Commit $B$}}
  \State \hspace{1.0em} \textbf{end if}
  \State
  \State \textbf{upon} receive $rankmsg$ \textbf{do}
  \State \hspace{0.8em} \textbf{if} \textsc{verify}($rankmsg$) $\wedge rankmsg.rank>curRank.rank$
  \State \hspace{2.0em} $curRank.rank \leftarrow rankmsg.rank$  
  \State \hspace{2.0em} $curRank.QC \leftarrow rankmsg.QC$
  \State \hspace{1.0em} \textbf{end if}
  \end{algorithmic}
\end{algorithm}

\subsubsection{Consensus Instance} \label{subsec:pbft}
We now describe \sysname using PBFT~\cite{pbft1999} for consensus instances (called \sysname-PBFT).

\bheading{Data structure.}
Messages are tuples of the form $\langle type$, $v$, $n$, $d$, $i$, $rank \rangle_{\sigma}$, where $type \in$ \{{\textsc{pre-prepare}}, {\textsc{prepare}}, {\textsc{commit}}, {\textsc{rank}}\}, $v$ indicates the view in which the message is being sent, $n$ is the round number, $d$ is the digest of the client’s transaction, $i$ is the instance index, $rank$ is the rank of the message, $\langle msg \rangle_\sigma$ is the signature of message $msg$. We use $ppremsg, premsg, commsg, rankmsg$ as the shorthand notation for pre-prepare, prepare, commit, and rank messages, respectively.
The instance index is added to mark which instance the message belongs to, since we run multiple instances of consensus in parallel. The parameter $rank$ is the MR, which is used for the global ordering of the blocks.  An aggregation of $2f+1$ signatures of a message $msg$ is called a Quorum Certificate ($QC$) for it. 
When we say a replica sends a signature, we mean that it sends the signed message together with the original message and the signer's identity.

\bheading{Normal-case operation.}
Algorithm~\ref{alg:pbftnew} shows the operation of \sysname protocol in the normal case without faults. Within an instance, the protocol moves through a succession of views with one replica being the leader and the others being backups in a view. The protocol runs in rounds within a view.  
An instance starts at view $0$ and round $1$  and a unique instance index $i$.
The leader starts a three-phase protocol ({\textsc{pre-prepare}},{\textsc{prepare}},{\textsc{commit}}) to propose batches of transactions to the backups. After finishing the three phases, replicas commit the batch with corresponding parameters. 
We generally use a {\textsc{verify}} function to check the validity of a message, such as the validity of the signature and whether the parameters match the current view and round. 


\iheading{1) \textsc{pre-prepare}.} This phase is only for the leader. In round $n \neq 1$, upon receiving $2f+1$ $rankmsg$ (including one from itself) for round $n$ in round $n-1$, the leader proposes a batch for the new round. When a leader proposes for round $n$,
it forms a set $rankSet$ of the $rankmsg$ for round $n$ (Line 2), and picks the maximum rank value with its QC from $rankSet$, denoted as $rank_m$ and $QC$ (Line 3).
Then, the leader cuts a batch $txs$ of transactions (Line 4), and calculates the digest of $txs$ (Line 5). A rank number $rank=rank_m + 1$ is assigned to $txs$, which should not exceed the $maxRank(e)$ of the current epoch (Line 6). The leader multicasts a pre-prepare message (Line 8) with $txs$, $QC$, and $rankSet$ to all the backups, where $QC$ is proof for the validity of $rank_m$. The set $rankSet$ is used to prove the leader follows the rank calculation policy. After proposing a batch with the $maxRank(e)$, the leader stops proposing (Lines 9-10). Note that when $n = 1$, the leader doesn't need to wait for $rankmsg$, but let $rankSet[n] \leftarrow \langle${\textsc{rank}}$, v, n-1, \bot, i, curRank.rank\rangle_\sigma$.

\iheading{2)  {\textsc{prepare}}.} A backup accepts the pre-prepare message after the following validity checks:
\begin{packeditemize}
    \item The pre-prepared message meets the acceptance conditions in the original PBFT protocol.
    \item $rankSet$ contains $2f+1$ ($n \neq 1$) or $1$ ($n = 1$) signed $rankmsg$ from different replicas in current view and previous round (\ie, $rankmsg.v=v,rankmsg.n=n-1$).  
    \item If $rank_m$ is the highest rank in $rankSet$ and $rank_m \neq minRank$, $QC$ is a valid aggregate signature of $2f+1$ signatures for $rank_m$.
    \item If $rank_m +1 \leq maxRank(e)$, $rank=rank_m +1$; Otherwise, $rank=maxRank(e)$.
\end{packeditemize}

The backup then enters the prepare phase by multicasting a $\langle${\textsc{prepare}}$, v, n, d,$ $i, rank\rangle_{\sigma}$ message to all other replicas (Lines 15-16). Otherwise, it does nothing.

\iheading{3) {\textsc{commit}}.} Upon receiving $2f+1$ valid prepare messages from different replicas (Lines 19-20), a replica multicasts a commit message to other replicas (Line 22). If the $rank$ carried in the commit message is greater than the current highest rank that the replica knows $curRank.rank$, the replica updates its $curRank$ by setting $curRank.rank$ to $commsg.rank$ (Line 24), and generates a $QC$ for it by aggregating the $2f+1$ $premsg$ (Line 25). Then, a backup sends a $rankmsg$ together with the QC for $rank$ to the leader to report its $curRank$ (Lines 27-28). 

Finally, upon receiving $2f+1$ valid commit messages, a replica commits a block $B$ with its corresponding parameters (Lines 31-34). All the committed blocks will be globally confirmed and delivered to clients. 

\iheading{4) Respond to clients.}
When a replica commits a block, it checks whether any undelivered block can be globally confirmed (see~\secref{subsec:gor}). If so, it assigns the block a global index $sn$ and delivers it back to clients.

\bheading{View-change mechanism.}
If the leader fails, an instance uses the PBFT view-change protocol to make progress~\cite{pbft1999}.
A replica starts a timer for round $n+1$ when it commits a batch in round $n$ and stops the timer when it commits a batch in round $n+1$. If the timer expires in view $v$, the replica sends a view-change message to the new leader. After receiving $2f+1$ valid view-change messages, the new leader multicasts a new-view message to move the instance to view $v+1$.  
Thereafter, the protocol proceeds as we described it.

\subsection{Message Complexity Refinement} \label{subsec: optimization}
{We note that the leader must broadcast at least $2f+1$ rank messages to allow backups to authenticate the accuracy of the leader's $rank$ calculation. This prevents Byzantine leaders from arbitrarily selecting a $rank$ for the new proposal. 
However, this results in a communication complexity of $O(n^2)$ in the pre-prepare phase (which is $O(n)$ in PBFT). 
We provide an optimization of \sysname-PBFT using aggregate signature schemes to reduce its message complexity. The optimized protocol is referred to as \sysname-opt.

\bheading{High-level Description.} Our key idea is to aggregate the $2f+1$ rank messages into one by using the aggregate signature scheme.
Recall that standard multi-signatures or threshold signatures require that the same message be signed~\cite{nocommit}. 
This is, however, not the case for us, because different replicas may have different $rank$ values.
Instead, rather than encoding the $rank$ value information in the message directly, we encode it in the private keys. 

We modify the rules for generating rank messages as follows. 
For each replica, we generate $K$ private keys. In round $n$, when replica $r$ creates a rank message, it computes the difference between the highest rank that it knows(denoted as $rank_r$) and the rank of the current round (denoted as $rank$), \ie, $k \leftarrow (rank_r - rank)$. The replica then signs the message using its $k$th signature key, \ie, $rankmsg \leftarrow \langle${\textsc{rank}} $v, n, \bot, i, rank \rangle_{\sigma_{r_k}}$. This scheme allows each replica to sign the same message. Upon receiving a rank message, the leader can recover the $rank_r$ intended to be transmitted by a replica $r$ by computing the sum of $rank$ and $k$. {Note that the rank difference could be beyond the number of private keys. We use the $K$th private key to sign the replica’s rank difference when the difference is beyond this bound.} The $K$ can be adjusted according to stragglers in a real deployment.  

\bheading{Detailed protocol.} We present an optimized version of normal-case operation for \sysname-PBFT in round $n \neq 1$. There are three modifications.  

\iheading{1)} \textsc{Pre-prepare}. Upon receiving $2f+1$ rank messages, the leader aggregates the partial signatures into a single signature $\sigma$, which the replicas can efficiently verify using the matching public keys. The $rankSet$ is set to the signature $\sigma$ instead of a set of $2f+1$ rank messages, which reduces the communication complexity from $O(n^2)$ to $O(n)$. The leader obtains the maximum $k$ from the $2f+1$ rank messages, denoted as $k_m$, and lets $ppremesg.rank = k_m + rankmsg.rank +1$.

\iheading{2)} \textsc{Prepare}. The backups will check the validity of $\sigma$: (a) it is a valid aggregate signature of $2f+1$ signatures from different replicas; and, (b) $ppremesg.rank = k_m + rankmsg.rank +1$, where $k_m$ is the maximum $k$ in $\sigma$.

\iheading{3)} \textsc{Commit}. A replica calculates the difference between the highest rank it knows and the $rank$ of the current round as $k$, \ie, $k \leftarrow curRank.rank- commsg.rank$. {If $k < K$, the replica then sends the  $rankmsg \leftarrow \langle${\textsc{rank}} $v,n,\bot, i, commsg.rank \rangle_{\sigma_{r_k}}$ together with $curRank.QC$ to the leader. Otherwise, the replica signs the $rankmsg$ with the $K$th private key.}
Upon receiving a rank message, a replica updates its $curRank$ if $rankmsg.rank + rankmsg.k > curRank.rank$.

\subsection{Correctness Analysis Overview} \label{subsec:correct}
In this section, we provide a brief security analysis of \sysname. Due to space constraints, we leave detailed proofs to Appendix \ref{appen:correct}. We show that \sysname satisfies \textit{totality}, \textit{agreement}, and \textit{liveness} properties. 

\bheading{Proof sketch.} 
To establish totality, we first prove that all partially committed blocks will eventually be globally confirmed. If a replica globally confirms a block $B$, it must have partially committed it. According to SB-Agreement and SB-Termination, all honest replicas will partially commit $B$. Therefore, all honest replicas will globally confirm $B$.
To show agreement, we use induction and proof-by-contradiction. The proof relies on two key observations. First, each block is assigned a unique global ordering index, ensuring a one-to-one mapping between blocks and global ordering indexes. Second, if two honest replicas globally confirm different blocks with the same global ordering index, it directly contradicts the established protocol rules, highlighting that blocks with the same global ordering index must be identical. This straightforward reasoning establishes the uniqueness of the global ordering index and ensures the agreement property.
Regarding liveness, our bucket rotation mechanism ensures that all transactions will eventually be assigned to an honest leader for processing. Then we prove that a transaction proposed by an honest leader will be eventually partially committed and then globally confirmed.

\section{Evaluation} \label{sec:evaluation}
In this section, we evaluate the performance and causality of \sysname across different scenarios. We compare \sysname against four state-of-the-art Multi-BFT protocols: ISS~\cite{stathakopoulou2022state}, RCC~\cite{gupta2021rcc},  Mir~\cite{MIR-BFT}, and DQBFT~\cite{dqbft}. 
We implemented \sysname in Go\footnote{{https://github.com/eurosys2024ladon/ladon}} and used the Go BLS library for aggregate signatures. The evaluation results are illustrated using ChiPlot\footnote{https://www.chiplot.online/}. 
We build two end-to-end prototypes of \sysname-PBFT and \sysname-HotStuff. 
Due to space constraints, we present the results of \sysname-PBFT in this section and leave the results of \sysname-HotStuff to Appendix \ref{appen:hotstuff}. 
For brevity, we refer to \sysname-PBFT as \sysname in this section.
Our experiments aim to answer the following research questions:
\begin{packeditemize}
    \item \textbf{Q1:} How does \sysname perform with varying number of replicas in WAN and LAN environments as compared to ISS,  RCC, Mir and DQBFT? (\secref{sec:expperformance})   
        
    \item \textbf{Q2:} How does \sysname perform under faults? (\secref{sec:expstraggler})
    
    
    \item \textbf{Q3:} How does the causal strength of \sysname compare to that of ISS, Mir, RCC, and DQBFT? (\secref{sec:expsecurity})
\end{packeditemize}
\begin{figure*}[t]
    \centering
    \begin{subfigure}[t]{0.24\textwidth}
        \centering
        \includegraphics[width=\textwidth]{ 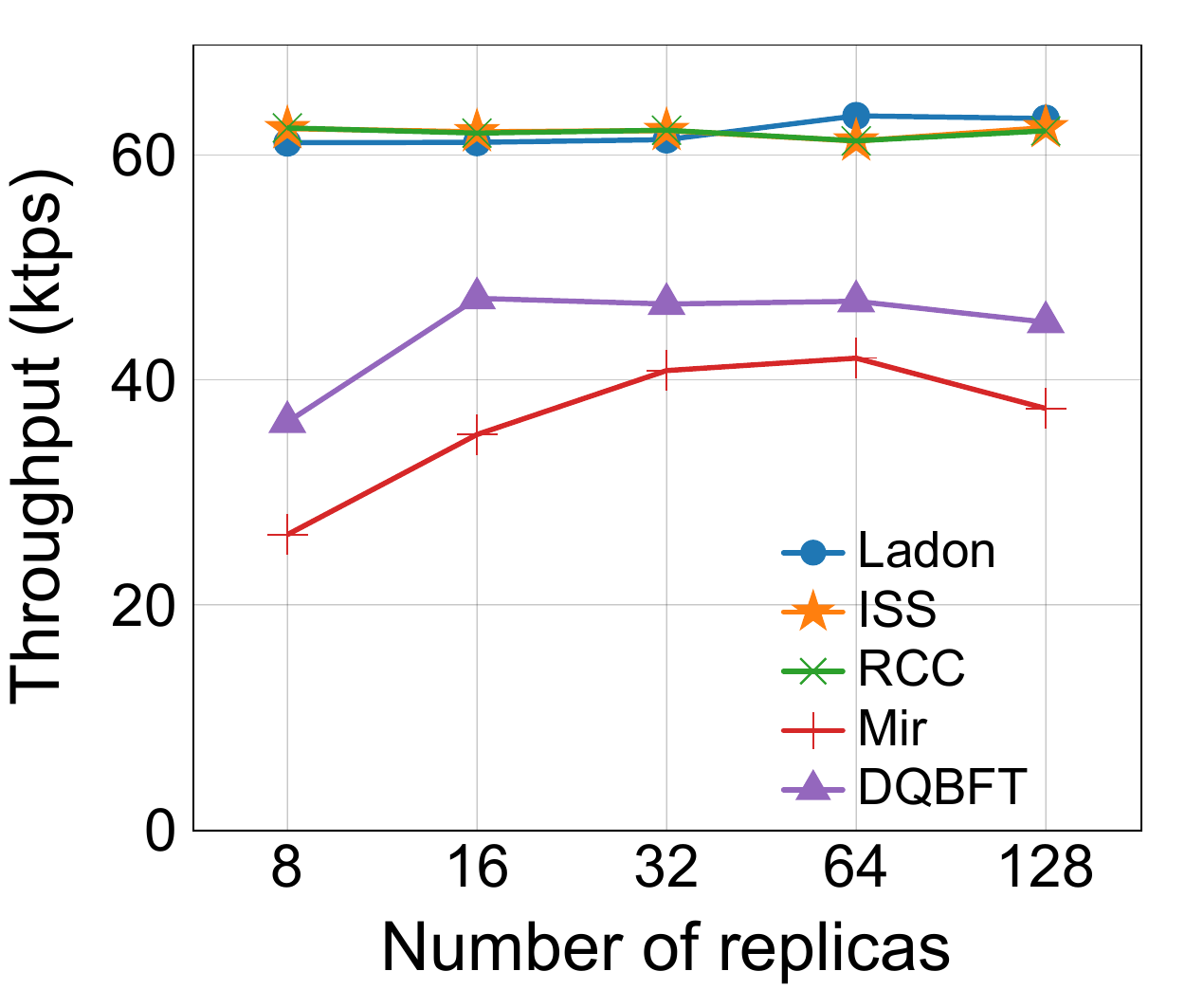}
        \caption{\#Straggler=0, WAN}
        \label{fig:wan2}
    \end{subfigure}
    \hfill
    \begin{subfigure}[t]{0.24\textwidth}
        \centering
        \includegraphics[width=\textwidth]{ 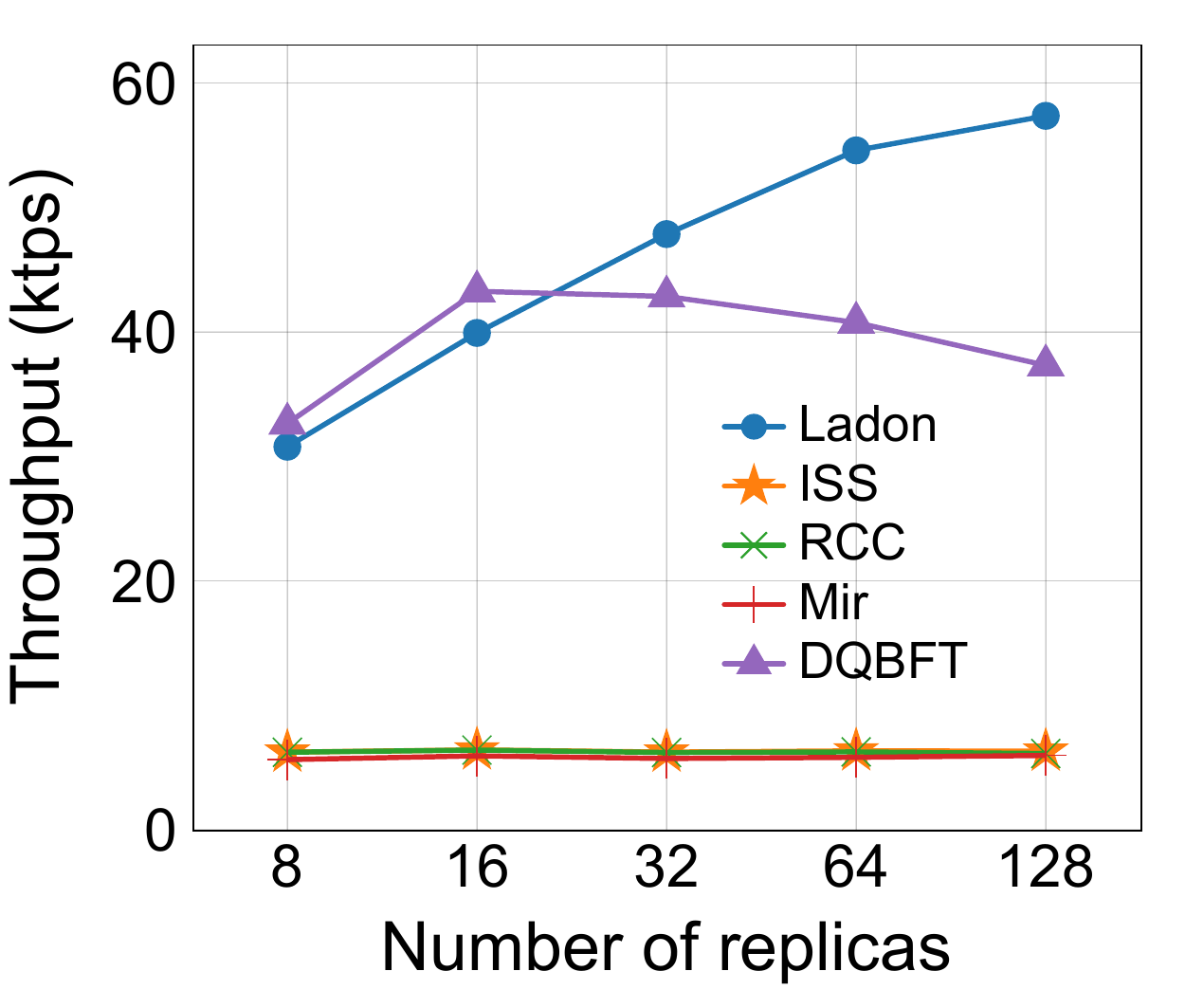}
        \caption{\#Straggler=1, WAN}
        \label{fig:wan1}
    \end{subfigure}
    \hfill
    \begin{subfigure}[t]{0.24\textwidth}
        \centering
        \includegraphics[width=\textwidth]{ 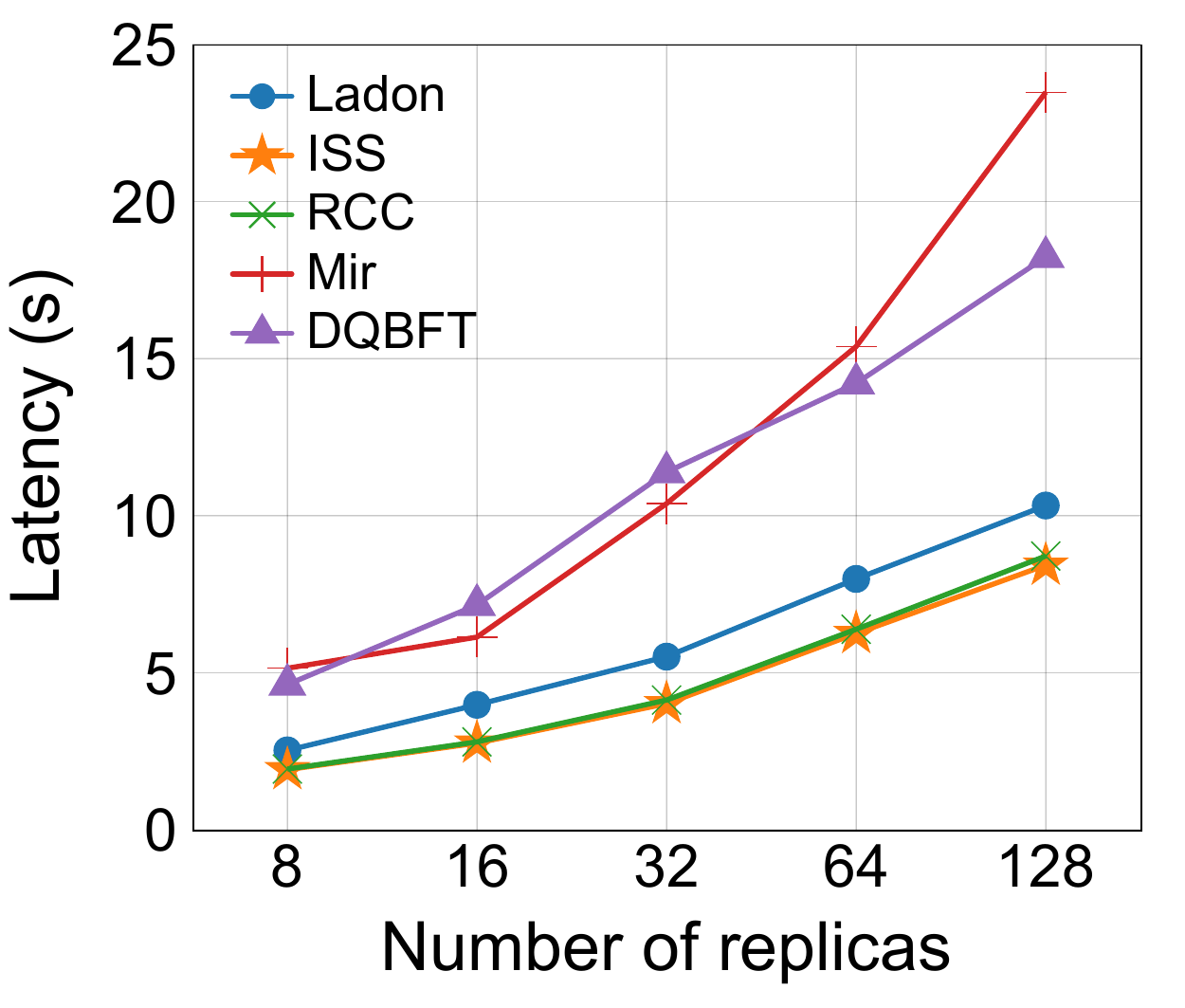}
        \caption{\#Straggler=0, WAN}
        \label{fig:wan4}
    \end{subfigure}
    \hfill
    \begin{subfigure}[t]{0.24\textwidth}
        \centering
        \includegraphics[width=\textwidth]{ 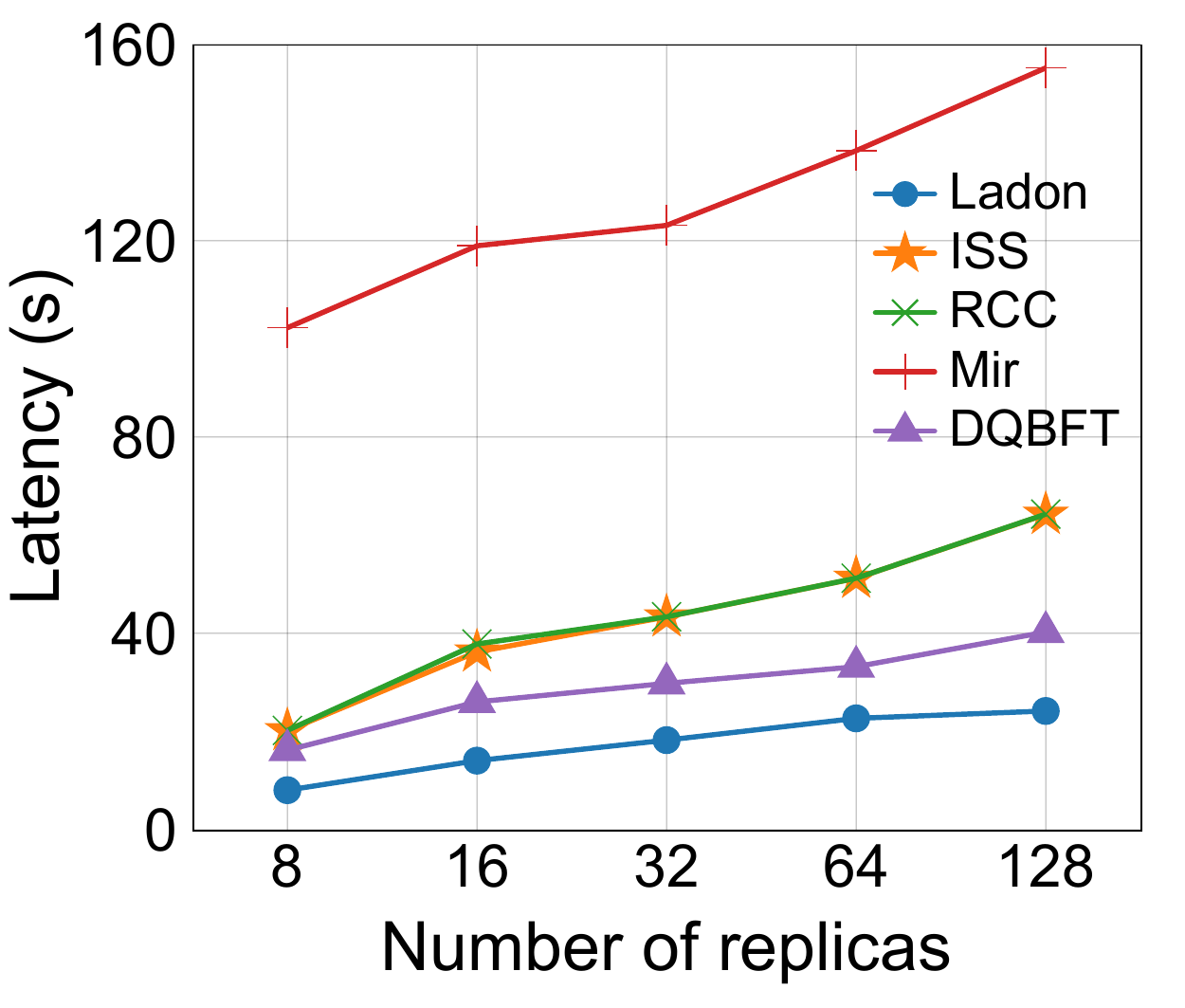}
        \caption{\#Straggler=1, WAN}
        \label{fig:wan3}
    \end{subfigure}
    \hfill

    \begin{subfigure}[t]{0.24\textwidth}
        \centering
        \includegraphics[width=\textwidth]{ 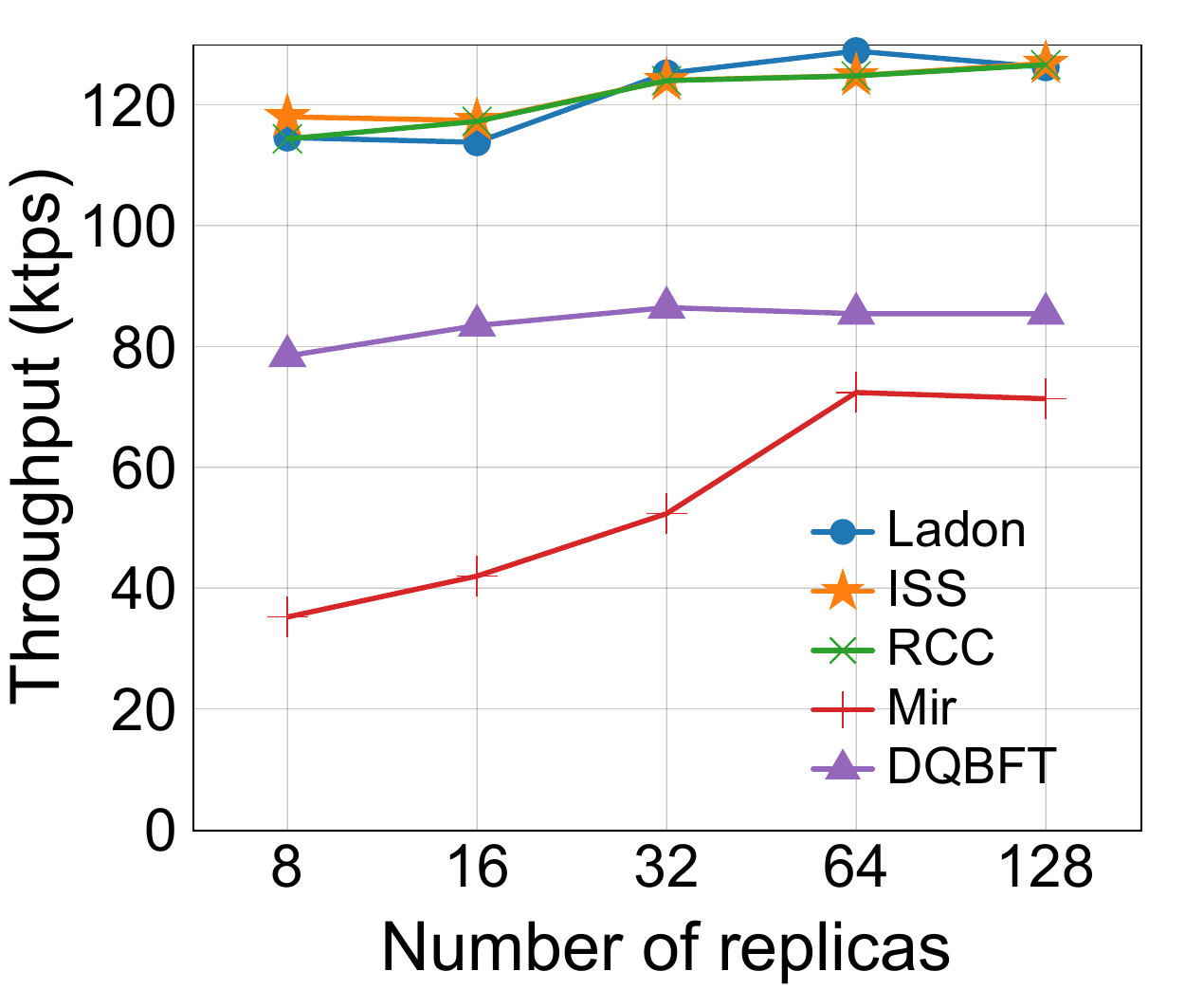}
        \caption{\#Straggler=0, LAN}
        \label{fig:lan2}
    \end{subfigure}
    \hfill
\begin{subfigure}[t]{0.24\textwidth}
        \centering
        \includegraphics[width=\textwidth]{ 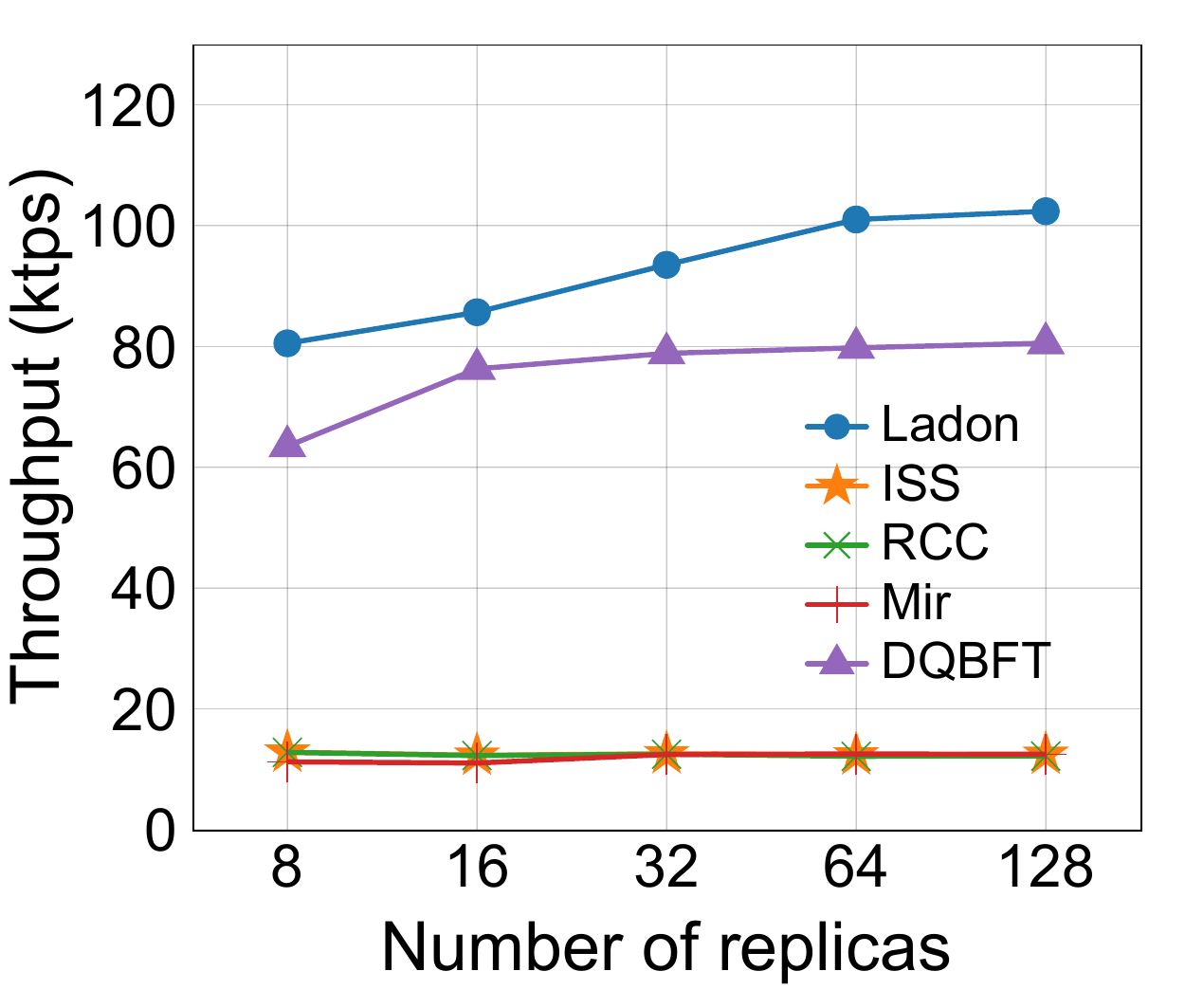}
        \caption{\#Straggler=1, LAN}
        \label{fig:lan1}
    \end{subfigure}
    \hfill
    \begin{subfigure}[t]{0.24\textwidth}
        \centering
        \includegraphics[width=\textwidth]{ 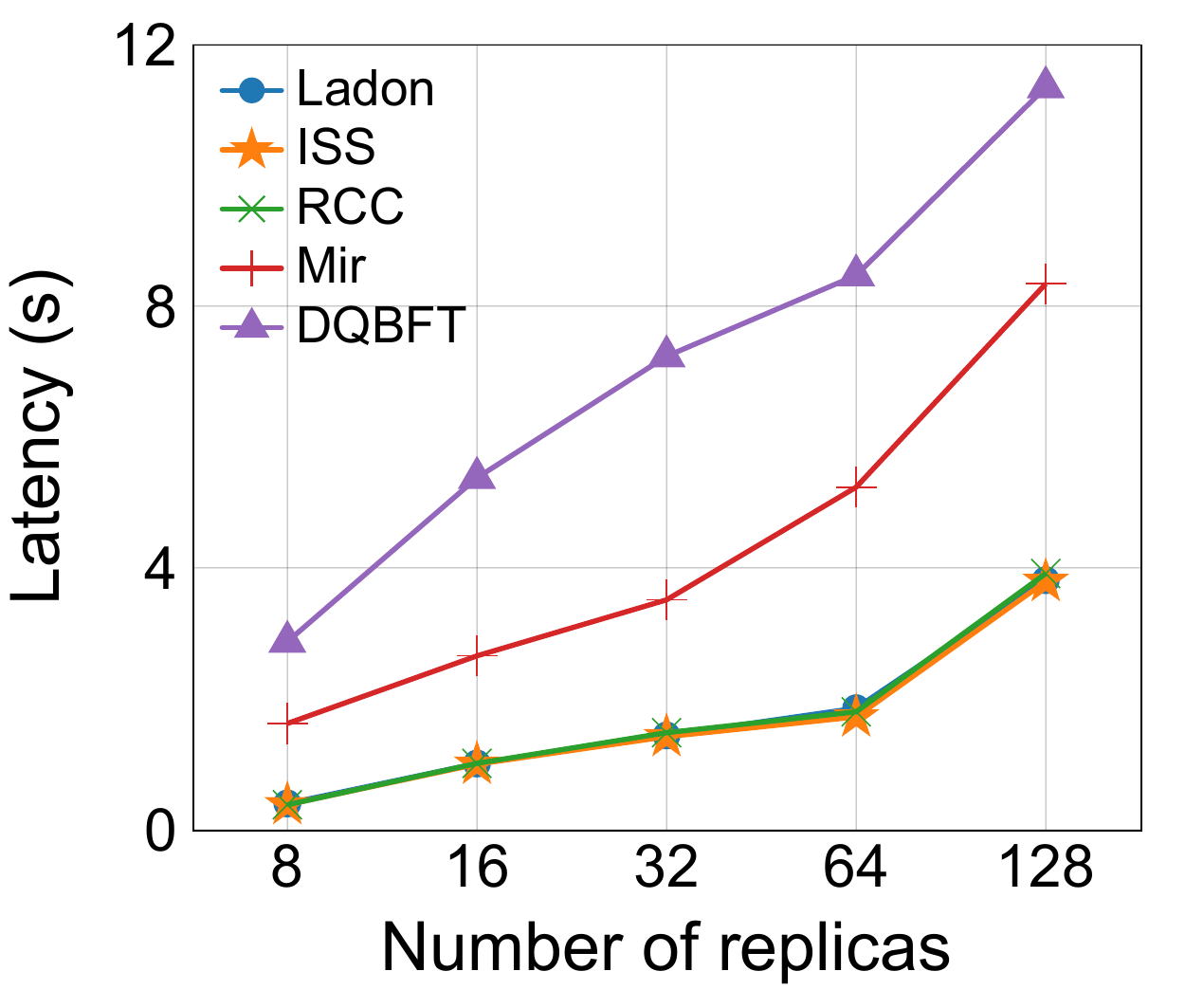}
        \caption{\#Straggler=0, LAN}
        \label{fig:lan4}
    \end{subfigure}
    \hfill
    \begin{subfigure}[t]{0.24\textwidth}
        \centering
        \includegraphics[width=\textwidth]{ 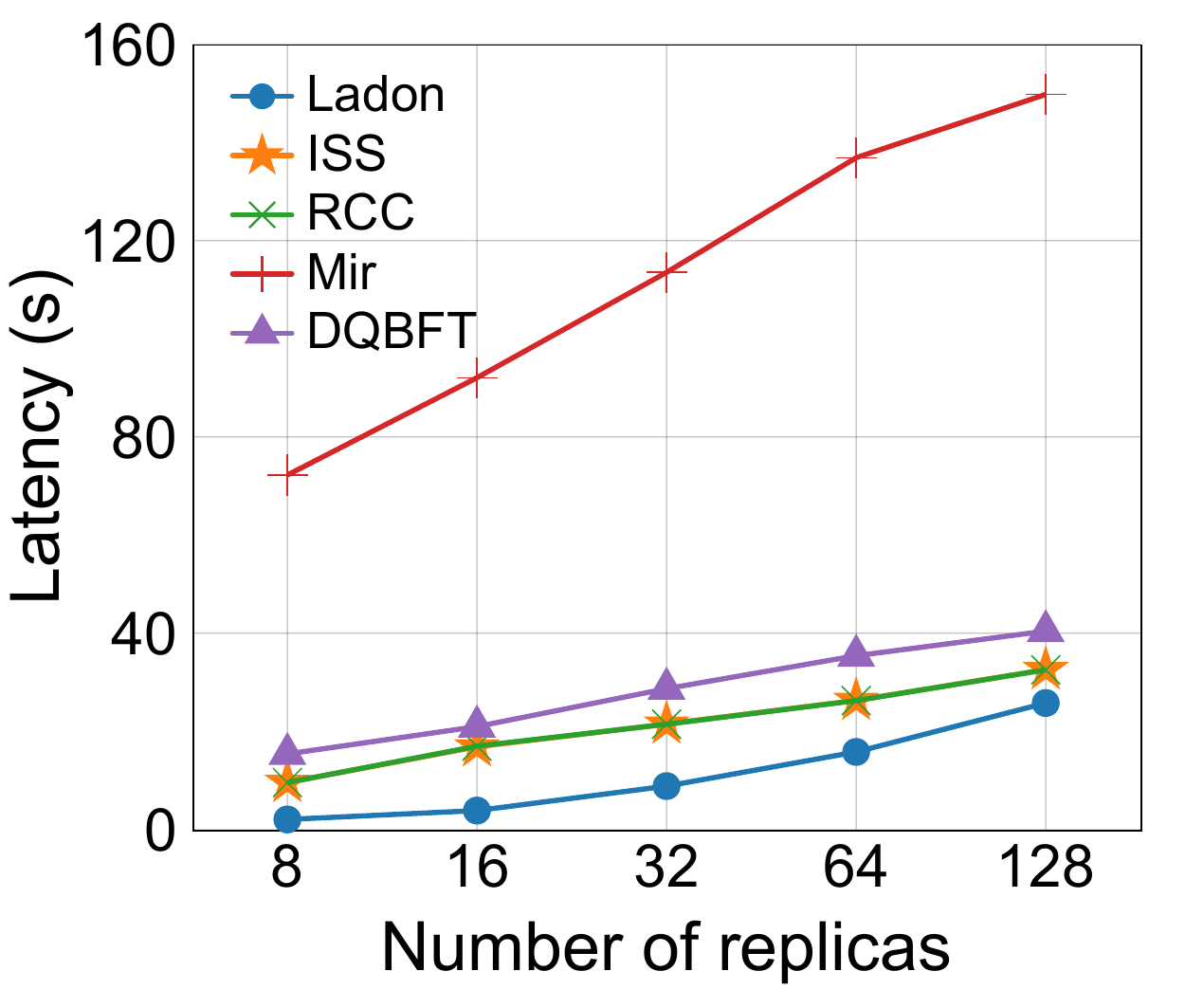}
        \caption{\#Straggler=1, LAN}
        \label{fig:lan3}
    \end{subfigure}

    \caption{\textbf{Throughput and latency of \sysname, ISS, RCC, Mir, and DQBFT  in WAN (a) -- (d), and LAN (e) -- (h).}}
    \label{fig:performance}
\end{figure*}

\subsection{Experimental Setup}\label{sec:expset}
\bheading{Deployment settings.} 
We deploy all systems on AWS EC2 machines with one c5a.2xlarge instance per replica. All
processes run on dedicated virtual machines with 8vCPUs and 16GB RAM running Ubuntu Linux 22.04.
We conduct extensive experiments of \sysname in LAN and WAN environments. For LAN, each machine is equipped with one private network interface with a bandwidth of 1Gbps. For WAN, machines span 4 AWS data centers across France, America, Australia, and Tokyo. We distribute the replicas evenly across the four regions. 
Each machine is equipped with a public and a private network interface. We limit the bandwidths of both to 1 Gbps. For WAN experiments, we use the public interface for client transactions and the private interface for BFT consensus. We use NTP for clock synchronization across the servers. We conduct 5 experimental runs for each condition and plot the average value.

\bheading{System settings.} 
To ensure a fair comparison, we employed the same system configuration across all the protocols (\sysname, ISS, RCC, Mir,  and DQBFT). 
Each transaction carries a 500-byte payload, which is the same as the average 
transaction size in Bitcoin~\cite{nakamoto2012bitcoin}.
Each replica operates as a leader for one instance, and as backup replicas for other instances, \ie, $m = n$. 
{We follow ISS by limiting the total block rate (number of blocks proposed by all leaders each second) to $16~blocks/s$ in WAN and $32~blocks/s$ in LAN. This prevents the leader from trying to propose too many batches in parallel, which triggers a view change timeout. However, this constraint leads to a higher end-to-end latency as we increase the number of nodes.}
We allow a large batch size of $4096$ transactions.
The epoch length is fixed at $l(e)=64$ for both protocols. 

{While we acknowledge that the above setting may not be exhaustive or optimal, conducting exhaustive experiments to identify the optimal configuration exceeds the scope of this work. However, our chosen parameters enable us to demonstrate that \sysname outperforms other protocols in the presence of stragglers while introducing minimal overhead. This is a key contribution of our work. Here, we focus on discussing the most critical configuration parameters.}

\bheading{Straggler settings.} 
{We simulate two types of stragglers. The first type, evaluated in \secref{sec:expperformance} and \secref{sec:expsecurity}, are honest stragglers. They follow the ISS protocol, delaying proposals for a set period without triggering timeouts and not including transactions in their blocks. Stragglers are randomly selected, with proposal rates (number of blocks proposed by a leader each second) fixed to $1/k$ of normal leaders, where $k$ is a parameter.
The second type, evaluated in \secref{sec:expstraggler}, are Byzantine stragglers. These behave like honest stragglers and also manipulate rank selection by collecting more than $2f+1$ ranks, discarding the higher ranks, and using the lowest $2f+1$ ranks before proposing a new block.}

\subsection{Failure-Free Performance}\label{sec:expperformance}
We evaluate the performance of \sysname and its counterparts under two conditions: with one honest straggler and without stragglers in both WAN and LAN environments. We also evaluate the performance of \sysname and its counterparts  with a varying number of  honest stragglers in WAN.
In this section, we adopt $k=10$ for stragglers.
We measure the peak throughput in kilo-transactions per second (ktps) before reaching saturation along with the associated latency in second (s) in \secref{sec:expwan} and \secref{sec:explan}, and analyze the CPU and bandwidth usage in \secref{exp:cpu}. 
We define the throughput and latency as follows: 1) throughput: the number of transactions delivered to clients per second, and 2) latency: the average end-to-end delay from the time that clients submit transactions until they receive $f+1$ responses. 

%

\subsubsection{Performance in WAN}\label{sec:expwan}
\figref{fig:wan2} shows the throughput of each protocol without stragglers with varying numbers of replicas. 
\sysname demonstrates a comparable throughput with ISS and RCC, with a minimal difference of approximately 1\% on 128 replicas, which demonstrates that \sysname only incurs minimal overhead. Furthermore, we notice that \sysname consistently outperforms Mir and DQBFT. 
\figref{fig:wan1} shows the throughput with one honest straggler.  The plot illustrates the superior throughput of \sysname, with $9.1\times$, $9.4\times$, and  $9.6\times$ of ISS, RCC, and Mir, respectively, on 128 replicas. This is because dynamic global ordering mitigates the performance degradation caused by stragglers that are prevalent in pre-determined global ordering schemes.
DQBFT, another dynamic ordering protocol, shows comparable throughput with \sysname initially, but its throughput declines as the number of replicas increases.

Comparing \figref{fig:wan2} and \figref{fig:wan1}, it is evident that the throughput of pre-determined global ordering protocols (\ie, ISS, RCC, and Mir) significantly drops by 89.9\%, 90.1\%, and 84.1\% on 128 replicas, respectively, in the presence of one straggler. In contrast, dynamic global ordering protocols (\ie, \sysname and DQBFT) are less affected by stragglers, with throughput drops only by 9.3\% and 17.3\% on 128 replicas, respectively. 

\figref{fig:wan4} and \figref{fig:wan3} show the latency for each protocol without stragglers and with one straggler.  With one straggler, the latency of \sysname and DQBFT is $2.3\times$ and $2.2\times$ of that with no stragglers on 128 replicas, respectively. In contrast, the latency of ISS, RCC, and Mir with one straggler increases significantly compared to that with no stragglers, achieving $7.6\times$, $7.4\times$ and $6.6\times$ on 128 replicas, respectively.
Without stragglers, \sysname's latency is 22.6\% and 18.5\% higher compared to ISS and RCC on 128 replicas while much lower than Mir and DQBFT. With one straggler, \sysname shows the lowest latency. 
The latency of all protocols increases as the number of replicas grows. This phenomenon arises from our decision to maintain a fixed total block rate. Consequently, with more replicas, the time interval for each replica to propose a block becomes longer.

\begin{figure}[t]
	\centering
    \hfill
    \begin{subfigure}[t]{0.23\textwidth}
        \centering
        \includegraphics[width=\textwidth]{ 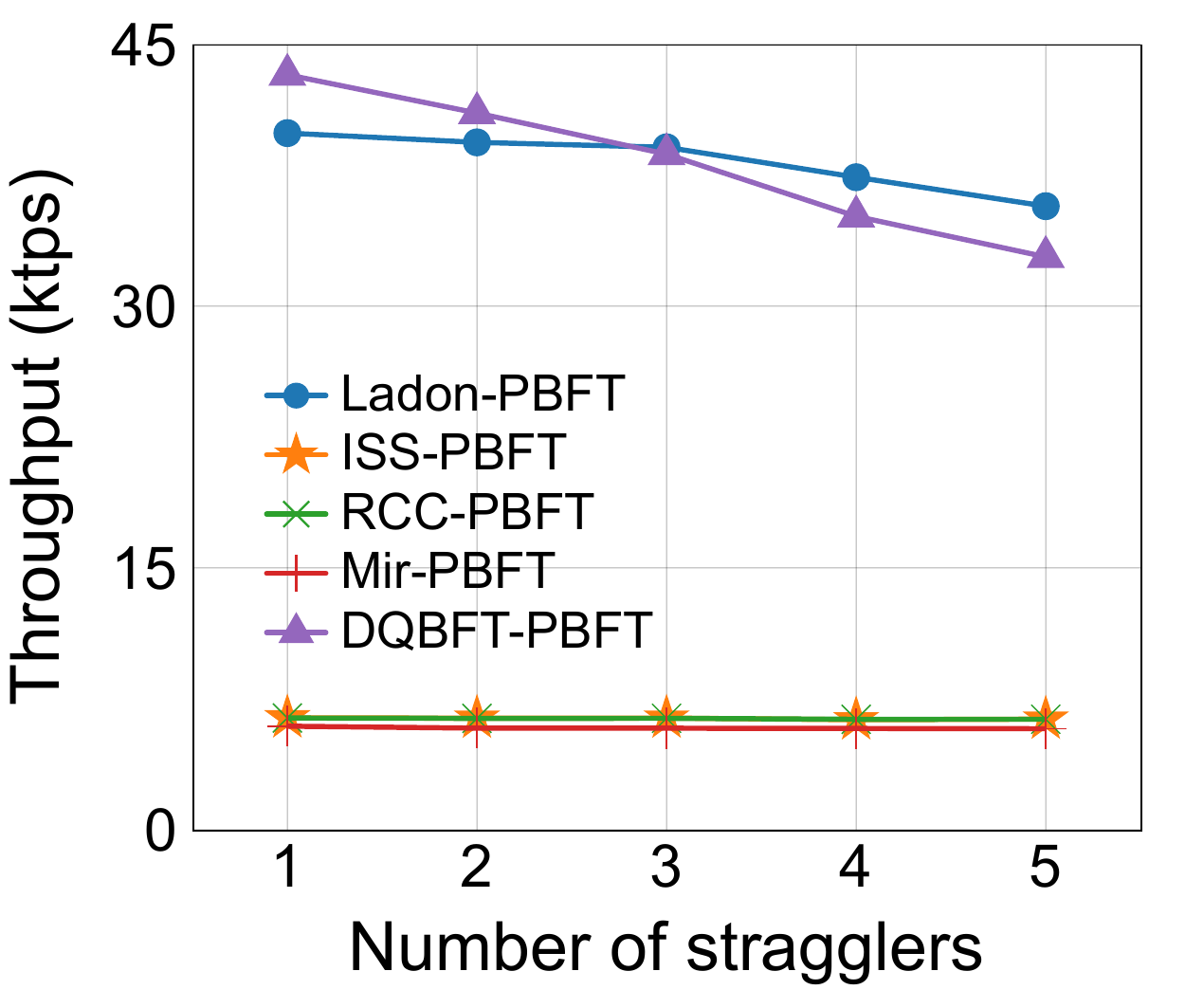}
        \caption{Throughput}
        \label{fig:throughputstragglers}
    \end{subfigure}
    \hfill
    \begin{subfigure}[t]{0.23\textwidth}
        \centering
        \includegraphics[width=\textwidth]{ 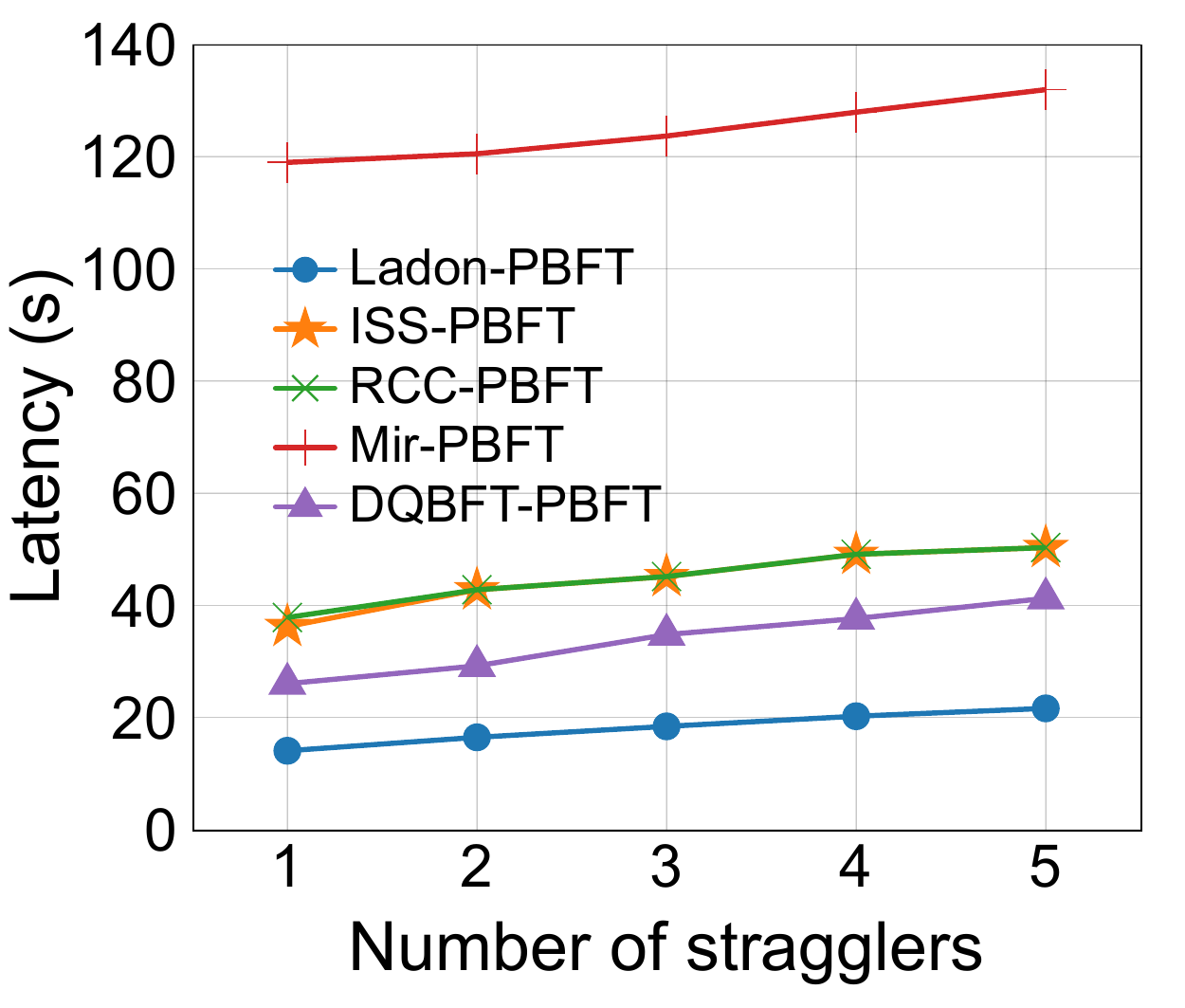}
        \caption{Latency}
        \label{fig:latencystragglers}
    \end{subfigure}
    \caption{\textbf{Throughput and latency of \sysname-PBFT and other protocols with a varying number of stragglers.}}
	\label{fig:12345s}
\end{figure}

\figref{fig:12345s} evaluates the performance of each protocol with a varying number of stragglers in WAN. We use 16 replicas in these experiments and vary the number of stragglers from 1 to 5. 
\figref{fig:throughputstragglers} shows that the throughput of \sysname, ISS,  RCC, Mir, DQBFT drop by 10\%, 1\%, 1\%, 2\% and 24\%, from one straggler to 5 stragglers, respectively. 
In \figref{fig:latencystragglers}, the latency increases slightly for all protocols. 
From \figref{fig:12345s}, we observe the robustness of these protocols against the rise in straggler count. The throughput and latency largely remain steady despite the increasing number of stragglers. 
This is because the system performance is limited by the slowest straggler, as discussed in \secref{subsec: straggler}.

\subsubsection{Performance in LAN}\label{sec:explan}

We evaluate the throughput and latency of \sysname and other protocols without stragglers and with one honest straggler in a LAN environment. The results are shown in \figref{fig:lan2}$-$\figref{fig:lan3}. All protocols exhibit similar performance trends to those observed in the WAN environment (\figref{fig:wan2}$-$\figref{fig:wan3}), with higher throughput and reduced latency.
Specifically, \sysname shows comparable performance with ISS and RCC without stragglers and always outperforms other protocols with one straggler.

\subsubsection{CPU and Bandwidth Analysis}\label{exp:cpu} 
To delve deeper into the performance bottlenecks of the protocols, we conducted an assessment of the CPU and bandwidth usage for both ISS and \sysname. The summarized results for 32 replicas in a WAN (16 blocks/s) and LAN (32 blocks/s) are presented in Table~\ref{tab:cpu}.  
Our findings indicate that neither ISS nor \sysname is constrained by CPU resources, as the maximum CPU usage possible is 800\%, given that each instance is equipped with 8 vCPUs. Without stragglers, \sysname exhibits comparable bandwidth usage to ISS.
With one straggler, \sysname generally experiences higher network bandwidth consumption and CPU usage compared to ISS. 

\begin{figure}[t]
	\centering
    \hfill
    \begin{subfigure}[t]{0.23\textwidth}
        \centering
        \includegraphics[width=\textwidth]{ 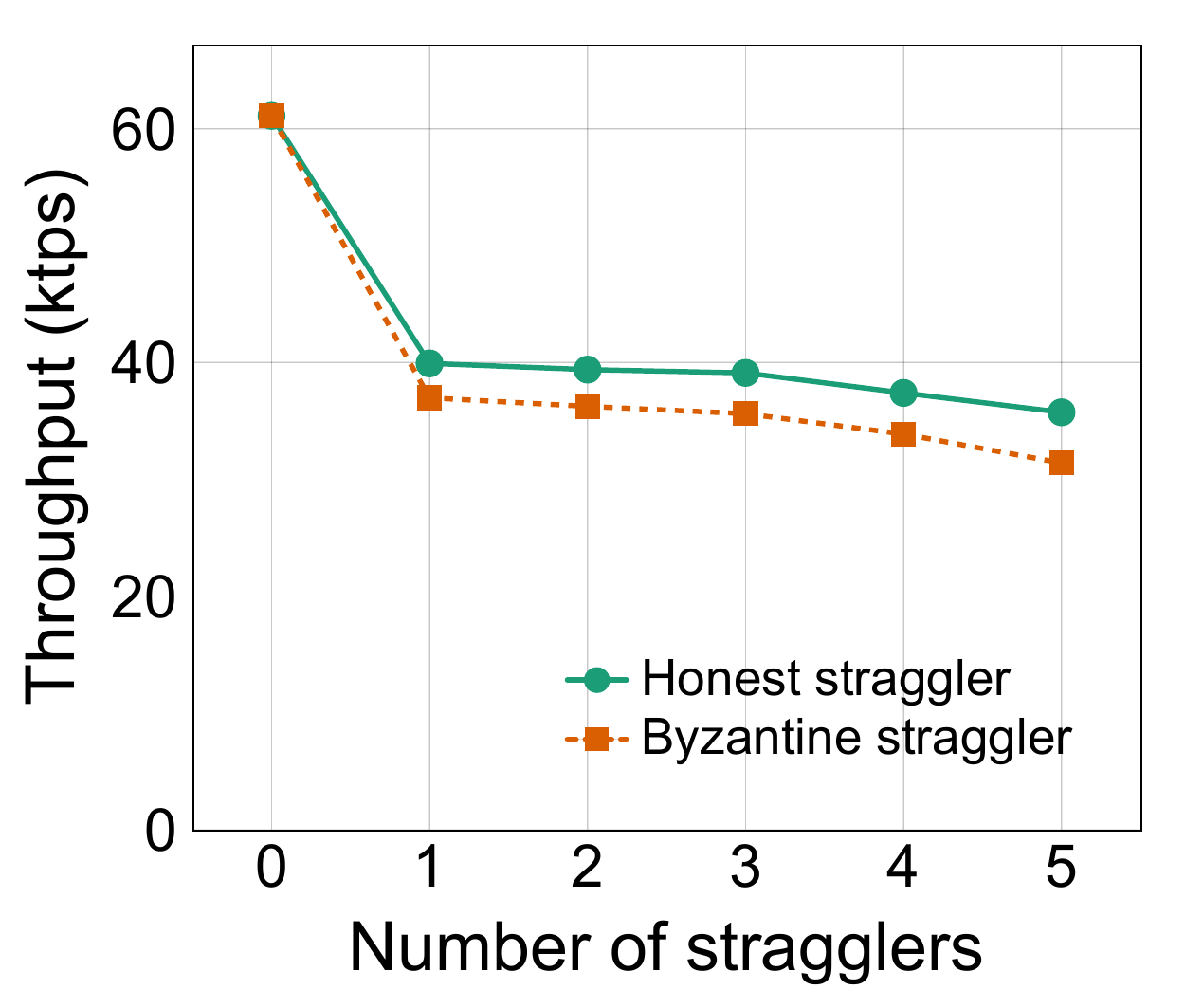}
        \caption{Throughput}
        \label{fig:byzantinestragglerthroughput}
    \end{subfigure}
    \hfill
    \begin{subfigure}[t]{0.23\textwidth}
        \centering
        \includegraphics[width=\textwidth]{ 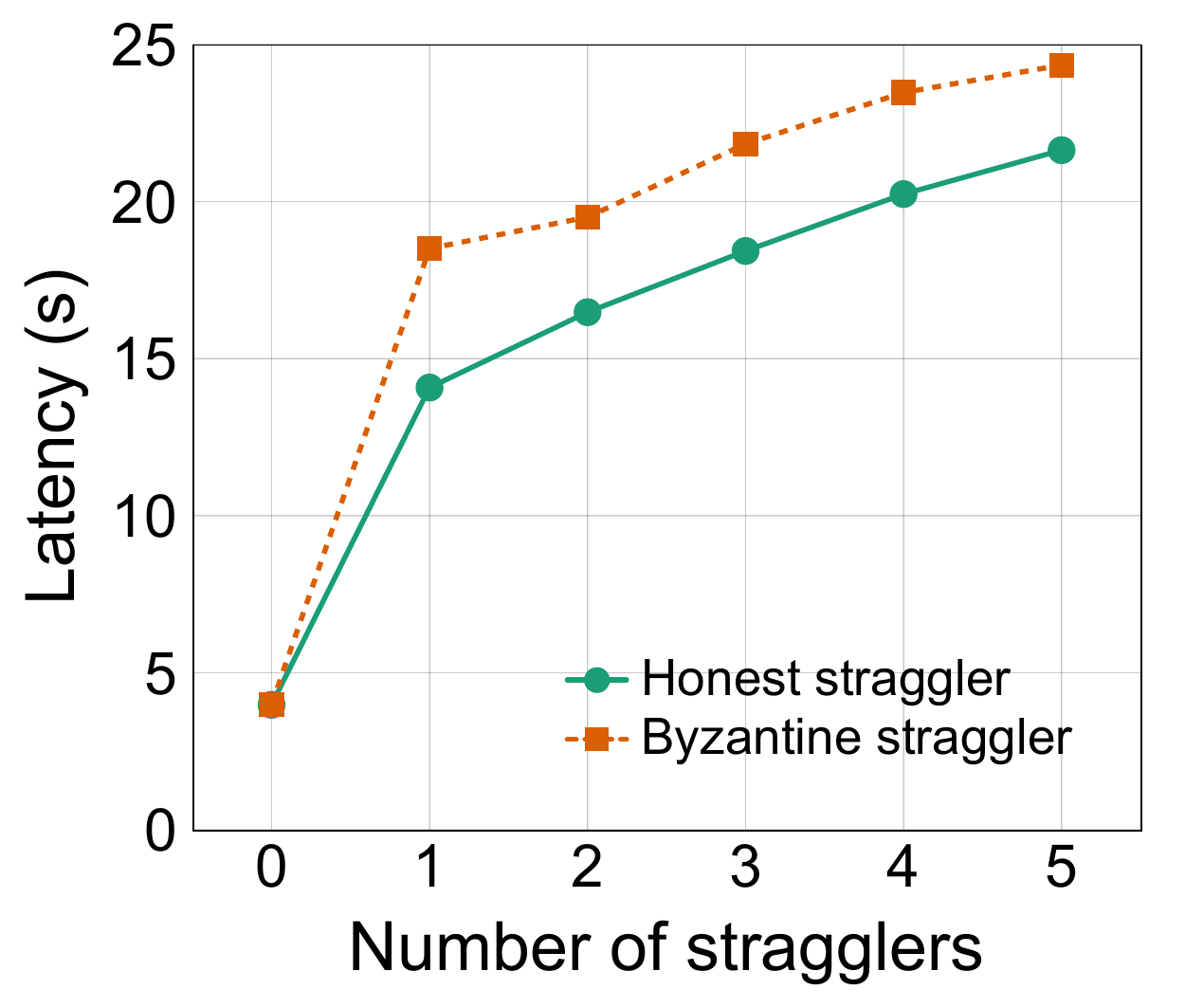}
        \caption{Latency}
        \label{fig:byzantinestragglerlatency}
    \end{subfigure}
    \caption{\textbf{Throughput and latency of \sysname with honest and Byzantine stragglers.}}
	\label{fig:byzantinestraggler}
\end{figure}
\begin{table}[t]
    \centering
    \caption{\textbf{CPU and bandwidth usage of \sysname and ISS. Maximum CPU usage possible is 800\%.}}
    \begin{tabular}{@{}lcccc@{}}\toprule
        Protocols \& Settings & Env & Block rate &CPU &Bandwidth\\ \midrule
        \multirow{2}{*}{ISS-0-stragglers} & WAN & 16 b/s & 319\% &85MB/s   \\
         & LAN & 32 b/s & 566\% &  160MB/s\\  \midrule
        \multirow{2}{*}{ISS-1-straggler}&  WAN & 16 b/s & 132\% & 25MB/s  \\
        & LAN & 32 b/s & 292\% & 77MB/s  \\
        \midrule
        \multirow{2}{*}{\sysname-0-stragglers} &  WAN & 16 b/s & 350\% &99MB/s   \\
         & LAN & 32 b/s & 591\% &  175MB/s\\  \midrule
        \multirow{2}{*}{\sysname-1-straggler} &  WAN & 16 b/s & 195\% & 54MB/s  \\ 
        &  LAN & 32 b/s & 432\% & 121MB/s  \\ 
        \bottomrule
    \end{tabular}
    \label{tab:cpu}
\end{table}

\begin{table*}[t]
    \centering
        \caption{\textbf{Causal Strength ($CS$) of different protocols for different numbers of stragglers and proposal rates.}
    }
    \label{tab:csn}
    \vspace{-0.5em}
    \setlength{\abovecaptionskip}{0pt} 
    \resizebox{\textwidth}{!}{%
    \begin{tabular}{@{}lcccccccccc@{}}\toprule
        \textbf{Stragglers} & \multicolumn{5}{c}{\textbf{\# Stragglers}} & \multicolumn{5}{c}{\textbf{Proposal rate (blocks/s)}} \\ \cmidrule(r){2-6} \cmidrule(r){7-11}
        \textbf{settings} & {1}  & {2} & {3} & {4} & {5} & {0.5}  & {0.4} & {0.3} & {0.2}& {0.1} \\ \midrule
        \textbf{Mir} & 0.154 & 0.042 & 0.012 & 0.004 & 0.002 & 0.241 & 0.204 & 0.174 & 0.148 & 0.154 \\ 
        \textbf{ISS} & $1.04 \times 10^{-5}$ & $3.42{\times 10^{-7}}$ & $6.79{\times 10^{-10}}$ &  $1.75{\times 10^{-13}}$ & $1.83{\times 10^{-16}}$ & ${0.078}$ & {$4.73\times 10^{-3}$} & 1.36$\times 10^{-4}$ & 7.28$\times 10^{-5}$ & 1.04$\times 10^{-5}$ \\
        \textbf{RCC} & $8.45 \times 10^{-6}$ & $2.48{\times 10^{-7}}$ & $5.44{\times 10^{-10}}$& $9.87{\times 10^{-14}}$ &  $1.18{\times 10^{-16}}$ & $0.076$ & {$2.49\times 10^{-3}$} & $9.88\times 10^{-5}$ & 5.07$\times 10^{-5}$ & 8.45$\times 10^{-6}$ \\
        \textbf{DQBFT} & $1.15 \times 10^{-5}$ & $2.17{\times 10^{-7}}$ & $9.74{\times 10^{-10}}$& $1.02{\times 10^{-13}}$ &  $6.85{\times 10^{-17}}$ & $0.044$ & {$7.51\times 10^{-3}$} & 9.15$\times 10^{-4}$ & 4.35$\times 10^{-4}$ & 1.15$\times 10^{-5}$ \\
        \midrule
        \textbf{\sysname} & 1.0 & 1.0 & 1.0 & 1.0 & 1.0 & 1.0 & 1.0 & 1.0 & 1.0 & 1.0\\
        \bottomrule
    \end{tabular}}
\end{table*}

\subsection{Performance Under Faults}\label{sec:expstraggler}
 
In this section, we study the performance of \sysname under Byzantine stragglers and crash faults in a WAN of 16 replicas. 

\subsubsection{Byzantine Stragglers}\label{subsec:expstraggler}
We study the impact on throughput and latency of Byzantine stragglers, with $f = 1$ up to the maximum tolerated number of $f = 5$ stragglers.
\figref{fig:byzantinestraggler} shows the impact of an increasing number of stragglers. \sysname with Byzantine stragglers reaches $\approx 90\%$ of the throughput with honest stragglers. 
The latency increases by 12.5\%  with 5 Byzantine stragglers compared to the latency with 5 honest stragglers.
These results indicate that the impact of Byzantine stragglers on the system's performance is only slightly more pronounced than that of honest stragglers. This is because the manipulation of ranks by Byzantine stragglers is limited, as discussed in~\secref{subsec:byzantinestraggler}.

\begin{figure}[t]
	\centering 
        \includegraphics[width=0.49\textwidth]{ 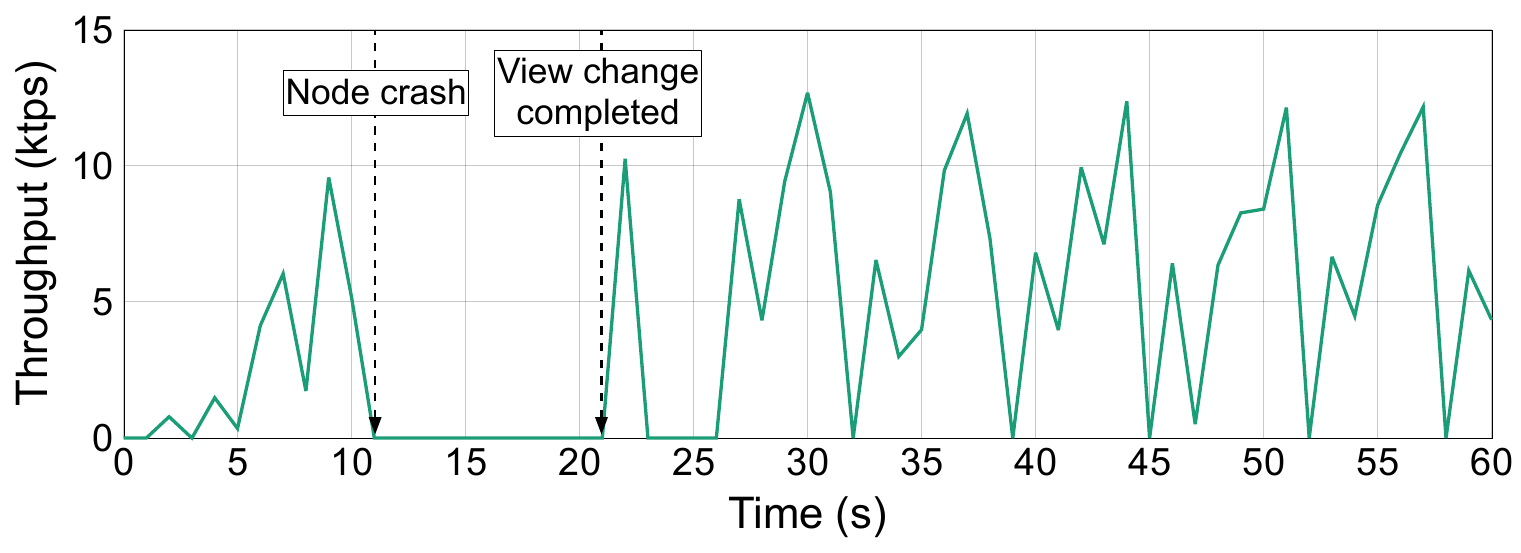}
    \caption{\textbf{\sysname's throughput average (over 1s intervals) over time with one crash fault. The crash is at 11s, and the view change is completed at 21s.}}
	\label{fig:crash}
\end{figure}

\subsubsection{Crash Failures}
We study how crash faults affect the throughput of \sysname. The PBFT view change timeout is set at $10$ seconds. \figref{fig:crash} shows the throughput average over time. A crash occurs in the first epoch at 11 seconds, causing the throughput to drop to 0. The view change process is initiated to handle the crash fault and is completed at 21 seconds, at which point the throughput begins to recover.  This delays the epoch change. A new epoch starts at 26 seconds. Each subsequent short drop to 0 in throughput corresponds to an epoch change.

\subsection{Causality Evaluation}\label{sec:expsecurity}
We first define a metric, Inter-block Causal Strength ($CS$) and then use it to evaluate the causality of \sysname.

\bheading{Inter-block Causal Strength $CS$.} Assume that a series of $n$ blocks $\{B_1$, $B_2$, $...$, $B_n\}$ has been globally confirmed. For any $i<j$, if $B_i$ is generated after $B_j$ is committed by $f+1$ replicas, we say a causality violation has occurred. The number of causality violations is denoted as $N$, and 
$CS = e^{-N/n}$.

The Inter-block Causal Strength ($CS$) is defined as a measure to evaluate the strength of causality of a system, with a value in (0, 1]. The closer the CS score is to 1, the stronger the system's causal property. A CS score of 1 implies that no one can front-run a partially committed block as discussed in~\secref{subsec:causality}.  On the other hand, as the number of causality violations increases, the system's causality weakens, driving the CS score closer to 0.


\bheading{Results.} 
We evaluate the $CS$ by varying the number of stragglers with a fixed proposal rate of 0.1 blocks/s, and by varying the proposal rate of a single straggler. We conducted experiments in a WAN environment with 16 replicas. 
We show results for \sysname-PBFT (short for \sysname in the following), which are similar to those of \sysname-HotStuff.

Table~\ref{tab:csn} shows that $\sysname$ always exhibits strong inter-block causality across all straggler settings.
By contrast, Mir, ISS, RCC, and DQBFT display a diminishing $CS$ with an increasing number of stragglers, reflecting a weakening of the system's causal property. The $CS$ of ISS, RCC, and DQBFT progressively decreases with the proposal rate.
This trend suggests that these protocols are susceptible to the presence of stragglers, which weaken their causal properties.
 
We attribute the difference in causal strengths in Table~\ref{tab:csn} to the different ordering mechanisms. Traditional BFT protocols utilize predetermined ordering in which even slow instances with straggling leaders might receive a global ordering for blocks earlier than other blocks in faster instances. This practice leads to causal violations, as it does not ensure that the global order aligns with the actual generation sequence, which is discussed in~\secref{subsec:causality}. Although DQBFT centralizes the global ordering, it fails to consider the causality between blocks.
By contrast, \sysname employs a dynamic ordering mechanism that respects the causality inherent in block generation. This design ensures that the global order of blocks corresponds to their actual generation sequence. 


\section{Related Work} \label{sec:related} 
Existing leader-based BFT protocols, such as Zyzzyva~\cite{kotla2007zyzzyva}, PBFT~\cite{pbft1999}, and HotStuff~\cite{hotstuff}, suffer from the leader bottleneck.
Many approaches to scaling leader-based BFT consensus have been proposed. These can be divided into three classes: parallelizing consensus, reducing committee size, and optimizing message transmission.

\bheading{Parallelizing consensus.} The idea in these approaches is that every replica acts as the leader to propose blocks, making all replicas behave equally. A representative method of this approach is Multi-BFT consensus~\cite{MIR-BFT, stathakopoulou2022state, gupta2021rcc}, which runs several consensus instances in parallel to handle transactions. Stathakopoulou \etal~\cite{MIR-BFT} propose Mir-BFT, in which a set of leaders run the BFT protocol in parallel. Each leader maintains a partial log, and all instances are eventually multiplexed into a global log.  To prevent malicious leaders, an epoch change is triggered if one of the leaders is suspected of failing. Byzantine leaders can exploit this by repeatedly ending epochs early to reduce throughput. Later, ISS~\cite{stathakopoulou2022state} improved on Mir-BFT, by allowing replicas to deliver $\perp$ messages and instances to make progress independently. This improved its performance in the presence of crash faults. RCC~\cite{gupta2021rcc} is another Multi-BFT protocol that operates in three steps: concurrent Byzantine commit algorithm (BCA), ordering, and execution. RCC adopts a wait-free mechanism to deal with leader failures, which does not interfere with other BCA instances. DQBFT~\cite{dqbft} adopts a special order instance to globally order the output transactions from other parallel instances. Nonetheless, the system performance undergoes a significant decline if the ordering instance has a straggling leader, and the centralization of the ordering process makes it a prime target for attacks. Multi-BFT consensus is simple and has high performance in ideal settings. However, as analyzed in \secref{sec:back&motivation}, Multi-BFT systems suffer from severe performance issues.

Similar to Multi-BFT consensus, DAG protocols~\cite{Narwhal-HotStuff, DAG-rider, Bullshark} also parallelize consensus instances to improve the system scalability. However, in DAG protocols, each block has to contain at least $2f+1$ references of predecessor blocks, instead of one reference in Multi-BFT consensus.  
We deliberately refrain from comparing our work with DAG-based systems due to fundamental differences in network assumptions and architectures, aligning with the standard practice in Multi-BFT research.

\bheading{Reducing committee size.} 
Another approach to improve BFT performance is to reduce the number of consensus participants, avoiding the leader bottleneck in large-scale settings. 
The representative solution is to randomly select a small group of replicas as the subcommittee, who are responsible for validating and ordering transactions. This solution has been developed by Algorand~\cite{Gilad2017}. 
The sharding approach takes one step further and divides replicas into multiple disjoint subcommittees. Subcommittees run BFT protocols in parallel to process clients' transactions, improving on the efficiency of a single subcommittee. 
Many BFT sharding protocols, such as Elastico~\cite{Luu2016}, OmniLedger~\cite{omniledger} and RapidChain~\cite{zamani2018rapidchain}, have been proposed. 
However, subcommittees lower the system's tolerance to Byzantine replicas (\eg, tolerating $25\%$ Byzantine replicas in Algorand rather than $33\%$).
The designs of these systems, especially state synchronization between subcommittees, are also more complex.  
Reducing committee size can significantly improve the scalability of BFT systems, however, it also weakens their fault tolerance and increases complexity.

\bheading{Optimizing message transmission.}
Substantial work has also gone into improving network utilization.
This line of work can be categorized according to the message transmission topology: structured and unstructured. 
A representative solution using unstructured transmission topology is gossip~\cite{buchman2018latest}, in which a replica sends its messages to some randomly sampled replicas.
Gossip has been used in Tendermint~\cite{Buchman2016TendermintBF}, Gosig~\cite{li2020gosig}, and Stratus~\cite{gai2023scaling}, which can remove the leader bottleneck in large-scale settings. 
By contrast, in a structured topology, each replica sends its messages to a fixed set of replicas. 
For example, in Kauri~\cite{10.1145/3477132.3483584}, replicas disseminate and aggregate messages over trees, and in Hermes~\cite{9543565} the leader sends blocks to a committee, which helps to relay the block and vote messages.
These approaches work well at scale (e.g., thousands of replicas) and have a high overhead at smaller scales, which are the focus of our work.

\section{Conclusion} \label{sec:conclusion}
We propose \sysname, a Multi-BFT protocol that mitigates the impact of stragglers on performance.
We propose dynamic global ordering to assign a global ordering index to blocks according to the real-time status of all instances, which differs from prior work using pre-determined global ordering. We decouple the dependencies between the various partial logs to the maximum extent to ensure a fast construction of the global log. We also design monotonic ranks, pipeline their distribution and collection with the consensus process, and adopt aggregate signatures to reduce rank data. We build \sysname-PBFT and \sysname-HotStuff prototypes, conducting comprehensive experiments on Amazon AWS to compare them with existing Multi-BFT protocols.
Our evaluation shows that \sysname has significant advantages over ISS, RCC, Mir, and DQBFT in the presence of stragglers. 



\bibliographystyle{plain}
\bibliography{main}

\appendix

\section{Message Complexity Analysis}\label{appen: optimization}
We focus on the message complexity of \sysname-PBFT and \sysname-opt in the normal case operations because the introduced dynamic global ordering does not affect the view-change mechanism. 
The message complexity measures the expected number of messages that honest replicas generate during the protocol execution~\cite{guo2020dumbo}.
Specifically, we take PBFT as a baseline to illustrate the overhead of dynamic global ordering in \sysname. 

\iheading{1) PBFT.} In the \textsc{Pre-prepare} phase, an honest leader broadcasts a message to all backup replicas, resulting in a complexity of $O(n)$. In the \textsc{Prepare} and \textsc{Commit} phases, which necessitate all-to-all communications among honest nodes, the complexity escalates to $O(n^2)$. Therefore, PBFT exhibits a message complexity of $O(n^2)$.

\iheading{2) \sysname-PBFT.} We analyze the message complexity of a consensus instance using PBFT in \sysname-PBFT. 
In the \textsc{Pre-prepare} phase, messages include a list containing $2f+1$ rank information, resulting in a $O(n^2)$ complexity. The \textsc{Prepare} phase aligns with PBFT, also reflecting $O(n^2)$ complexity. However, in the \textsc{Commit} phase, each backup communicates its highest rank to the leader, resulting in an additional all-to-one communication with a complexity of $O(n)$. Consequently, the \textsc{Commit} phase manifests a combined complexity of $O(n^2+n) \approx O(n^2)$. 

\iheading{3) \sysname-opt.} In the \textsc{Pre-prepare} phase, the $2f+1$ rank information, previously $O(n)$, is condensed into a singular $O(1)$ complexity, thereby reducing the overall complexity to  $O(n)$. The \textsc{Prepare} and \textsc{Commit} phases in the \sysname-opt remain consistent with \sysname-PBFT.

To sum up, \sysname-PBFT increases the message complexity of the \textsc{Pre-prepare} phase from $O(n)$ to $O(n^2)$.
But, both \sysname-PBFT and \sysname-opt do not increase the overall message complexity of the protocol. 

\iheading{Authenticator complexity.} In addition to message complexity, recent work~\cite{hotstuff} introduces a concept called authenticator complexity (also known as signature complexity) to quantify the number of cryptographic operations (such as verifications) done by replicas.
Since in \sysname, every message contains at least one signature, the authenticator complexity also reflects the message complexity. 
In the \textsc{Pre-prepare} phase of PBFT, each backup replica verifies $O(1)$ signatures, whereas, in \sysname-PBFT, it verifies an additional $O(n)$ signatures. 
By contrast, the \sysname-opt employs aggregate signatures to optimize this to $O(1)$. Likewise, during the \textsc{Commit} phase, the leader verifies an extra $O(n)$ signatures in \sysname.
Thus, considering the verification of $O(n)$ signatures in commit messages, the \textsc{Commit} phase manifests a combined complexity of $O(2n) \approx O(n)$. 

\begin{figure*}[ht]
    \centering 
    \includegraphics[width=0.9\textwidth]{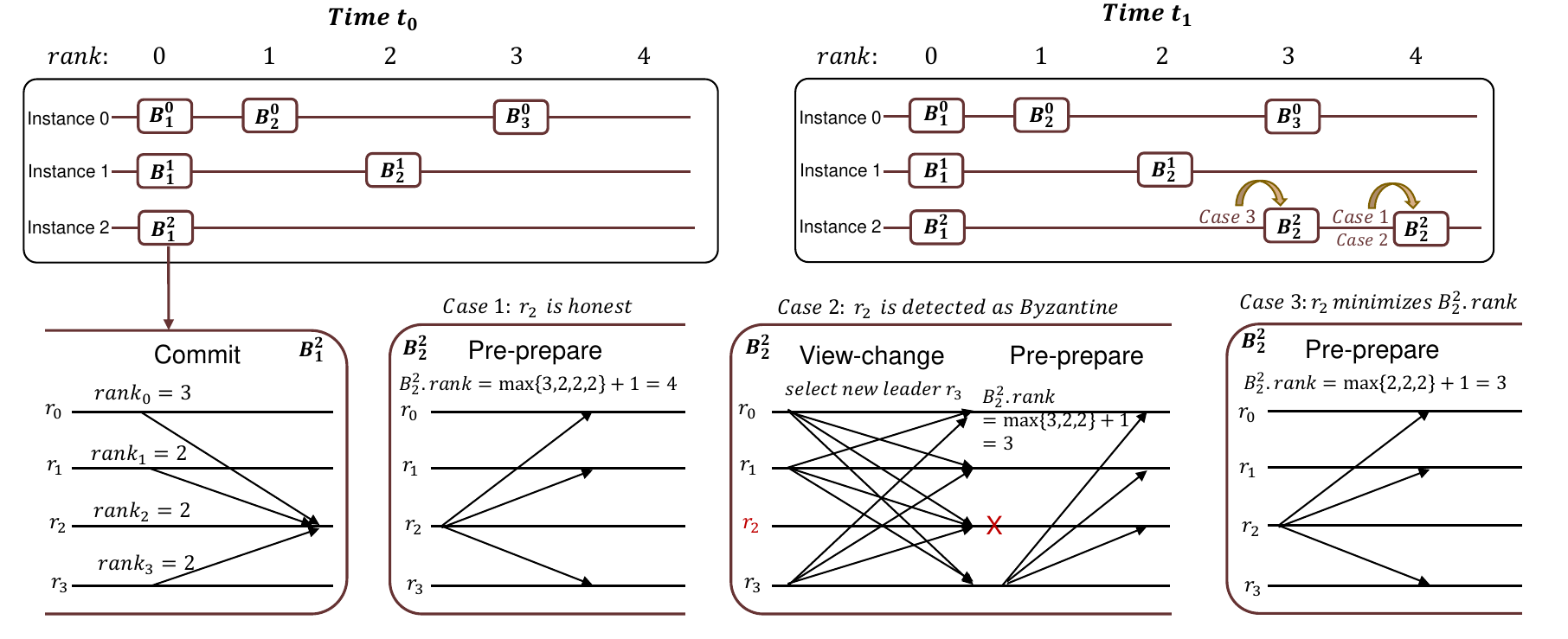}
    \caption{\textbf{Illustrative example of \sysname with four replicas and three instances under three cases:} (1) with an honest leader, (2) with a detectable Byzantine leader, and (3) with an undetectable Byzantine leader who minimizes the rank for the newly proposed block.}
    \label{fig:example}
\end{figure*}

\section{Illustrative Example of Protocol Behaviors}\label{appen:example}

Here we provide an illustrative example of \sysname protocol behaviors under the three cases: with an honest leader, with a detectable Byzantine leader, and with an undetectable Byzantine leader who minimizes the rank for the newly proposed block. 

Consider the example shown in \figref{fig:example}, where we have four replicas: $r_0$, $r_1$, $r_2$, and $r_3$. In this scenario, $r_0$, $r_1$, and $r_2$ are the leaders of Instance 0, 1, and 2, respectively. With $f = 1$, there are $3f + 1 = 4$ replicas in total.
During the commit phase of $B_1^2$, each replica $r_i$ sends its highest known rank, $rank_i$, to $r_2$. Suppose the ranks received by $r_2$ are $\{3, 2, 2, 2\}$. At time $t_1$, $r_2$ is about to generate a new block for Instance 2. 

\iheading{Case 1:} If $r_2$ is honest, it will select the highest rank from the set and add one, resulting in $B_2^2.rank = 4$.

\iheading{Case 2:} If $r_2$ is detected as a Byzantine leader, a view-change protocol will be initiated to replace $r_2$ with another replica as the leader of Instance 2, assumed here to be $r_3$, which is honest. Upon taking over, the new leader $r_3$ will collect the rank set $\{3, 2, 2\}$ from $r_0$, $r_1$, and itself, respectively. $r_3$ will then select the highest rank from the set and add one, resulting in $B_2^2.rank = 4$.

\iheading{Case 3:} If $r_2$ is Byzantine and wishes to avoid detection, it can discard the highest rank ($rank_0=3$ from $r_0$), as it only needs $2f + 1$ ranks to proceed. By selecting the highest rank from the remaining set $\{2, 2, 2\}$ and adding one, it would set $B_2^2.rank = 3$.

\section{Proof of Correctness} \label{appen:correct}

\subsection{Correctness of Monotonic Rank}
We prove that monotonic ranks in \sysname satisfy the two properties: MR-Agreement and MR-Monotonicity. 

\begin{lemma}\label{lemma:sbagree}
    If an honest replica partially commits a block $B$ then all honest replicas eventually partially commit $B$. 
\end{lemma}
\begin{proof}
     If an honest replica $r$ partially commits a block $B$, by SB-Termination, all honest replicas will partially commit a block for $B.round$. By SB-Agreement, if an honest replica $r'$ partially commits a block $b'$ for $B.round$, then $B'=B$. Thus,  all honest replicas eventually partially commit $B$.
\end{proof}

\begin{theorem}[MR-Agreement]\label{theorem:agreement}
All honest replicas have the same rank for a partially committed block.
\end{theorem}

\begin{proof}
When a leader proposes a block B, it determines block parameters and formats $B$ as $B=\langle txs, index, round, rank \rangle$. 
If an honest replica $r$ partially commits the block $B$, all other honest replicas eventually partially commit $B$ (Lemma~\ref{lemma:sbagree}). As $B.rank$ is one of the parameters of block $B$, all honest replicas have the same rank for block $B$. The proof is done.
\end{proof}

\begin{theorem}[MR-Monotonicity]\label{causality}
If a block $B'$ is generated after an intra-instance (or a partially committed inter-instance) block $B$, then the rank of $B'$ is larger than the rank of $B$.
\end{theorem}

\begin{proof} 
When block $B$ is partially committed inter-instance block, $B.rank$ is also agreed by a quorum $\mathcal{Q}$ of at least $2f+1$ replicas by the MR-Agreement property.
Later,  when a block $B^{\prime}$ is generated, its rank is set according to the highest ranks collected from a quorum $\mathcal{Q}^{\prime}$ of at least $2f+1$ replicas. Since there are $n=3f+1$ replicas in total, the intersection $\mathcal{Q} \cap \mathcal{Q}^{\prime}$ ensures that at least one honest replica in $\mathcal{Q}$ will report its highest rank $rank_m$, and $rank_m \geq B.rank$. By Algorithm~\ref{main}, $B^{\prime}.rank = rank_m +1$, which is larger than $B.rank$, and satisfies the monotonicity property. 

When block $B$ is an intra-instance block, according to Algorithm~\ref{alg:pbftnew}, the leader will collect $2f+1$ ranks for $B'$ at the commit phase of $B$, at which point the prepare phase has been finished and the leader has update its highest rank $curRank.rank$ if it is smaller than $B.rank$ (Lines 23-26). When propose $B'$, the leader will select the highest rank $rank_m$ from the $2f+1$ collected ranks, which contain the highest rank from itself. So we have $rank_m \geq B.rank$, $B^{\prime}.rank = rank_m +1 >B.rank.$, which satisfies the monotonicity property. 
\end{proof}

\subsection{System Security Properties}
\begin{lemma}\label{Monotonicity}
The $rank$s for blocks generated by the same BFT instance is strictly increasing, \ie, $B^i_{j+1}.rank > B^i_j.rank$. 
\end{lemma}

\begin{proof}
Assuming block $B^i_j$ is in epoch $e$.
By Algorithm~\ref{alg:pbftnew}, at the commit phase of $B^i_j$, an honest replica will update its highest rank value, denoted by $curRank.rank$, if $B^i_j.rank > curRank.rank$ (Lines 23-25). Then the replica will send its highest rank value $curRank.rank$ $\geq $ $B^i_j.rank$ to the leader (Line 27).
Upon proposing $B^i_{j+1}$, the leader will collect a set of $2f+1$ $rank$ values from different replicas (Line 1), denoted as $rankSet$, among which at least $f+1$ $rank$s are from those honest replicas. 
The highest rank that the leader gets from the $rankSet$ will be $rank_m \geq B^i_j.rank$.
By Line 6, if $rank_m+1 \leq maxRank(e)$, $B^i_{j+1}.rank = rank_m+1 > B^i_j.rank$; otherwise $B^i_{j+1}.rank = maxRank(e)$.

When $B^i_{j+1}.rank = maxRank(e)$, we discuss two different scenarios.
If $B^i_j.rank = maxRank(e)$, the leader will stop proposing for epoch $e$ after it proposes $B^i_j$, and $B^i_{j+1}$ will be in epoch $e+1$.  We have the $minRank(e+1)=maxRank(e)+1$, \ie, the minimum rank of a block in the $e+)$ epoch is one more than the maximum rank of the blocks in the epoch $e$, thus $B^i_{j+1}.rank > B^i_j.rank$.
If $B^i_j.rank < maxRank(e)$, $B^i_{j+1}.rank = maxRank(e) > B^i_j.rank$.

In summary, $B^i_{j+1}.rank > B^i_j.rank$.
\end{proof}

\begin{lemma}[]\label{lemmap2g}
If an honest replica partially commits a block $B$, it will eventually globally confirm $B$.
\end{lemma}

\begin{proof}
If an honest replica $r$ partially commits a block $B$ in epoch $e$, it will be added to the input of the partial ordering layer $\mathcal{G}_{in}$ (Algorithm~\ref{alg:pbftnew}, Line 34), which triggers Algorithm~\ref{main}. 

In Algorithm~\ref{main}, a replica $r$ first fetches the last partially confirmed block
from each instance,  denoted by the set $\mathcal{S'}$ (Line 2), and finds the block with the lowest ordering index among the blocks in $\mathcal{S'}$, denoted by $B^{*}$ (Line 3). Then $r$ computes the $bar$ as a tuple of $(B^{*}.rank + 1,B^{*}.index)$ (Line 4). Thereafter, $r$ will transverse all the blocks that have been partially committed by not globally confirmed yet, and global confirms all the blocks that have a lower ordering index than the $bar$ (Lines 5-11).
Thus, $B$ will be globally confirmed if it has a lower ordering index than the $bar$. 

Now we prove that $B$ will eventually have a lower ordering index than the $bar$.
If the leader of an instance is Byzantine, it will be detected by backups, and a view-change protocol is triggered to change the leader. Since there are more than $2/3$ honest replicas, there will eventually be an honest leader for the instance. 
If the leader of an instance is honest, 
it will continuously output partially committed blocks with increasing rank values (Lemma~\ref{Monotonicity}). 
Since $B^{*}$ represents the block with the lowest ordering index among the last partially confirmed blocks from each instance, eventually $B^*.rank \geq B.rank$. Thus the $rank$ of the $bar$ will be greater than $B.rank$, which means $B$ will have a lower ordering index than the $bar$ (In \sysname, blocks are ordered by increasing $rank$ values).

In summary, block $B$ will eventually be globally confirmed by $r$.
\end{proof}



\begin{theorem}[Totality]\label{totality}
If an honest replica globally confirms a block $B$, then all honest replicas eventually globally confirm the block $B$.
\end{theorem}

\begin{proof}
If an honest replica globally confirms $B$, it must have partially committed $B$.
By Lemma~\ref{lemma:sbagree}, all honest replicas eventually partially commit $B$. By Lemma~\ref{lemmap2g}, all honest replicas eventually globally confirm $B$.
\end{proof}

In the following text, we build on the concept of a global ordering index $sn$. In Algorithm~\ref{main}, global ordering refers to the act of appending a block to $\mathcal{G}_{out}$, which functions as a global log. The global ordering index $sn$ indicates the position of the block in the log, starting at $0$. With the addition of each block, the global ordering index $sn$ grows by one.

\begin{lemma}\label{uniquesn}
An honest replica uniquely maps a block $B$ to a single global ordering index $sn$.  
\end{lemma}
    
\begin{proof}
We first prove a block $B$ will be assigned a single global ordering index.
In Algorithm~\ref{main}, the set $\mathcal{S}$ represents the set of blocks waiting to be globally ordered (Line 5). A block will be moved out of $\mathcal{S}$ once it is globally ordered (Line 9). Thus, one block will only have one global ordering index.

Next, we prove that a single global ordering index corresponds to only one block.
Assuming an honest replica partially commits two blocks $B$ and $B'$. In Algorithm~\ref{main}, if $B$ and $B'$ are both in the set $\mathcal{S}$ (Line 5), only one of them can become the candidate block during a single round of searching (Line 6) and subsequently be appended to the global log (Line 8). The latter block will be assigned a greater global ordering index.
If $B'$  is added to $\mathcal{S}$ after $B$ has been globally confirmed, $B'$ will be appended to the global log at a later point and consequently receive a greater global ordering index than $B$.
Therefore, $B$ and $B'$ will have distinct global ordering indexes.

This concludes the proof, demonstrating that a single global ordering index uniquely corresponds to a specific block in the global log of an honest replica.
\end{proof}

\begin{theorem}[Agreement]
If two honest replicas globally confirm $B.sn = B'.sn$, then $B=B'$. 
\end{theorem}

\begin{proof}
We prove this statement by induction, considering the global block sequence from global ordering index $0$ to $sn$ for two honest replicas $r$ and $r'$, denoted as $L(sn)$ and $L'(sn)$, respectively. Here, $L(-1)$ signifies the initial state of the block sequence with no blocks.

First, we establish the base case $L(-1) = L'(-1) =\emptyset$. Then we provide the inductive step. Assuming $L(sn-1)=L'(sn-1)$, we aim to show that $L(sn)=L'(sn)$.

Assume $r$ and $r'$ globally confirm two different blocks $B$ and $B'$, respectively, with the same global ordering index $sn$. 
Since replica $r$ globally confirms $B$ with $sn$, we have 1) $r'$ will globally confirm $B$ (Theorem~\ref{totality}); 2) $B \notin L(sn-1)$ (Lemma~\ref{uniquesn}). Thus $B \notin L'(sn-1)$.
Now, if $r'$ globally confirms $B$ with the global ordering index $sn'$, by Lemma~\ref{uniquesn},  we know $sn' \neq sn$, hence $sn' > sn$. 
Similarly, if $r$ globally confirms $B'$ with the global ordering index $sn''$, we have $sn'' > sn$. 

According to the global ordering rules, for replica $r'$, if $sn'' > sn$, $B.rank > B'.rank$ or  $B.rank = B'.rank$ and $B.index > B'.index$; for replica $r$, if $sn' > sn$, $B'.rank > B.rank$ or  $B'.rank = B.rank$ and $B'.index > B.index$. 
However, this leads to a contradiction to the assumption that $B \neq B'$. This concludes the proof, demonstrating that if two honest replicas globally confirm $B.sn = B'.sn$, then $B=B'$ holds true.
\end{proof}

\begin{lemma}\label{lemmareq}
If a correct client broadcasts a transaction $tx$, some honest replica eventually proposes $B$ with $tx \in B.txs$.
\end{lemma}
\begin{proof}
Assuming $tx$ is assigned to buckets $\hat{b}$, and $\hat{b}$ is assigned to instance $i$. If the leader of the instance $i$ refuses to propose $tx$, $tx$ will be left in $\hat{b}$. According to the bucket rotation policy~\cite{stathakopoulou2022state}, $\hat{b}$ will be reassigned to different instances an infinite number of times.
If the leader of an instance is Byzantine, it will be detected by backups, and a view-change protocol is triggered to change the leader. Since there are more than $2/3$ honest replicas, there will eventually be an honest leader for the instance.
Thus, $\hat{b}$ will be eventually assigned to an instance with an honest leader, who will propose a block containing $tx$. 
\end{proof}

\begin{theorem}[Liveness]
If a correct client broadcasts a transaction $tx$, an honest replica eventually globally confirms a block $B$ that includes $tx$.
\end{theorem}
\begin{proof}
If a correct client broadcasts a transaction $tx$, by Lemma~\ref{lemmareq}, some honest replica eventually proposes $B$ with $tx \in B.txs$. By SB-Termination, the replica eventually partially commits $B$. By Lemma~\ref{lemmap2g}, it eventually globally confirms $B$.
\end{proof}

\section{\sysname with Chained HotStuff} \label{appen:hotstuff}
We now describe \sysname with the chained HotStuff~\cite{hotstuff} consensus instances (\sysname-HotStuff). 

\subsection{Data Structures}
\bheading{Tree and branches.} Each replica stores a tree of pending batches of transactions from clients as its local data structure. Each tree node contains a proposed batch, metadata associated with the protocol, and a parent link. 
The branch led by a given node is the path from the node all the way to the tree root by visiting parent links.
Two branches are conflicting if neither one is an extension of the other. Two nodes are conflicting if the branches
led by them are conflicting.

\bheading{Messages.} 
The general messages have a format of the tuple $\langle\langle type$, $view$, $index$, $rank$, $node$, $parentQC\rangle_{\sigma}$, $rank_m, rankQC\rangle$, where $type \in$ \{{\textsc{generic}}, {\textsc{vote}}, {\textsc{new-view}}\}
, $node$ contains a proposed node (the leaf node of a proposed branch), $parentQC$ is used to carry the $QC$ for the parent node, $rank_m$ and $rankQC$ are the highest rank known by the message sender and the QC for it, respectively. $\langle msg \rangle_\sigma$ is the signature of message $msg$. 
We use $genmsg$, $votemsg$, and $nvmsg$ as shorthand for the generic messages, vote messages, and new-view messages, respectively.

\begin{algorithm}
\caption{The \sysname-HotStuff Algorithm for Instance $i$
at View $v$ and Epoch $e$} 
\label{ChainedHotStuff}
\begin{algorithmic}[1]

\renewcommand{\algorithmicrequire}{\hspace{1.8em}\textcolor{purple}{\Comment{as a leader}}}
\Require

\State \textbf{upon} receive $2f+1$ $votemsg$ \textbf{do} 
\State \hspace{1.0em} $voteSet \leftarrow$ $2f+1$ $votemsg$
\State \hspace{1.0em} $QC \leftarrow$ {\textsc{generateQC}}$(voteSet)$
\State \hspace{1.0em} $txs \leftarrow$ {\textsc{cutBatch}}$(ins.bucketSet)$
\State \hspace{1.0em}$proposal \leftarrow$ {\textsc{createLeaf}}$(QC.node, txs)$
\State \hspace{1.0em}$rank \leftarrow min\{curRank.rank +1, maxRank(e)\}$
\State \hspace{1.0em}$genmsg \leftarrow \langle\langle${\textsc{generic}}, $v$,  $i$, $rank$,  $proposal$, $QC\rangle_{\sigma}$, $curRank.rank$, $curRank.QC$ $\rangle$
\State \hspace{1.0em} broadcast $\langle genmsg, voteSet \rangle$ 
\State \hspace{1.0em} \textbf{if} {$rank = maxRank(e)$}
\State \hspace{2.0em} stop propose
\State \hspace{1.0em} \textbf{end if}

\State
\renewcommand{\algorithmicrequire}{\hspace{1.8em}\textcolor{purple}{\Comment{as a replica}}}
\Require
\State \textbf{upon} receive $\langle genmsg, voteSet \rangle$ from leader \textbf{do}
\State \hspace{0em} \textbf{if} {\textsc{verify}}$(\langle genmsg, voteSet \rangle)$
\State \hspace{1.0em} \textbf{if} $genmsg.rank_m > curRank.rank$
\State \hspace{2.0em} $curRank.rank \leftarrow genmsg.rank_m$
\State \hspace{2.0em} $curRank.QC \leftarrow genmsg.rankQC$
\State \hspace{1.0em} \textbf{end if}
\State \hspace{1.0em} $B^* \leftarrow genmsg.node$ 
\State \hspace{1.0em} $B^{\prime \prime} \leftarrow B^*.QC.node$ 
\State \hspace{1.0em} $B^{\prime} \leftarrow B^{\prime \prime}.QC.node$ 
\State \hspace{1.0em} $B \leftarrow B^{\prime}.QC.node$
\State \hspace{0em} \textbf{end if}
\State \hspace{0em} \textbf{if} {\textsc{checkNode}} 
\State \hspace{0.8em} $votemsg$ $\leftarrow$ $\langle \langle genmsg\rangle_{\sigma}$, $curRank.rank,  curRank.QC\rangle$
\State \hspace{0.8em} send $votemsg$ to leader 
\State \hspace{0em} \textbf{end if}
\renewcommand{\algorithmicrequire}{\hspace{1.8em}\textcolor{purple}{\Comment{start {\textsc{precommit}} phase on $B^*$'s parent}}}
\Require
\State \hspace{0em} \textbf{if} $B^*.parent = B^{\prime \prime} $ 
\State \hspace{1.0em} $genericQC \leftarrow B^*.QC$
\State \hspace{0em} \textbf{end if}
\renewcommand{\algorithmicrequire}{\hspace{1.8em}\textcolor{purple}{\Comment{start {\textsc{commit}} phase on $B^*$'s parent}}}
\Require
\State \hspace{0em} \textbf{if} $(B^*.parent = B^{\prime \prime})$ $\wedge$ $ (B^{\prime \prime}.parent=B^{\prime})$ 
\State \hspace{1em} $lockedQC \leftarrow B^{\prime \prime}.QC$
\State \hspace{0em} \textbf{end if}
\renewcommand{\algorithmicrequire}{\hspace{1.8em}\textcolor{purple}{\Comment{start {\textsc{decide}} phase on $B^*$'s parent}}}
\Require
\State \hspace{0em} \textbf{if} $(B^*.parent = B^{\prime \prime})$ $\wedge$ $ (B^{\prime \prime}.parent=B^{\prime})$ $\wedge$ $B^{\prime}.parent$$=$$B$ 
\State \hspace{1.0em} $\mathcal{G}_{in} \leftarrow \mathcal{G}_{in} \cup  B$ \textcolor{purple}{{//Commit $B$}}

\State \hspace{0em} \textbf{end if}
\State
\renewcommand{\algorithmicrequire}{\hspace{1.8em}\textcolor{purple}{\Comment{as a leader}}}
\Require
\State \textbf{upon} receive $votemsg$ \textbf{do}
\State \hspace{0.5em} \textbf{if} {\textsc{verify}}$(votemsg) \wedge votemsg.rank_m > curRank.rank$
\State \hspace{2.0em} $curRank.rank \leftarrow votemsg.rank_m$
\State \hspace{2.0em} $curRank.QC \leftarrow votemsg.rankQC$
\State \hspace{0.5em} \textbf{end if} 
\end{algorithmic}
\end{algorithm}

\subsection{Algorithm Description}
As shown in Algorithm~\ref{ChainedHotStuff}, the Chained HotStuff consensus protocol simplifies the algorithm into a two-phase process: proposal and voting. 

\bheading{Proposal phase.} Upon receiving $2f+1$ $votemsg$ from the previous view (Line 1), a leader generates the QC for the last node using {\textsc{generateQC}} (Line 3) and proposes a new node extending the last node (Line 5). 
The leader sets the $rank$ of the new node as the highest rank it knows plus one. If the  $rank$ is greater than the $maxRank(e)$ of the current epoch, it is set to $maxRank(e)$ (Line 6). Then, the leader generates a $genmsg$ and broadcasts it to all replicas (Lines 7-18). If the $rank$ of the new proposal is equal to $maxRank(e)$, the leader stops proposing (Lines 9-10)\footnote{To ensure that all blocks can be delivered, we extend each partial log with 3 dummy blocks which are not added to the global log}.

\bheading{Voting phase.} Upon receiving the $genmsg$ from the leader, a replica validates it and updates its $currentRank$ if the $rank_m$ is greater than the highest rank it knows (Lines 15-17).
The replica then votes for it if it extends the highest node in the replica's view by sending a $votemsg$ to the leader (Lines 25-26). Finally, the replica will check the commit rule to decide whether there is a node that could be committed.  

\bheading{Commit rule.}
Chained HotStuff introduces a specific commit rule based on its chain structure: a node is committed when its 3-chain successor has received votes from a quorum of replicas. In the context of the Chained HotStuff commit rule, a 3-chain successor to a node refers to the node that is three positions ahead in the chain. For instance, if you have a chain of nodes $A \rightarrow B \rightarrow C \rightarrow D \rightarrow E$, node B is the 3-chain predecessor of node E.
For example, if a node B in the chain has a 3-chain successor E and E receives a quorum of votes, then B is committed.

\bheading{View change.} 
Similar to PBFT, HotStuff instantiates a view change when the current leader fails to propose a node within a certain period (this duration is typically predefined), or if the leader is suspected of being faulty. Replicas use a new-view message to carry the highest $genericQC$ and send it to the leader. 



\subsection{Performance Evaluation}
\begin{figure}[t]
	\centering
    \hfill
    \begin{subfigure}[t]{0.23\textwidth}
        \centering
        \includegraphics[width=\textwidth]{ 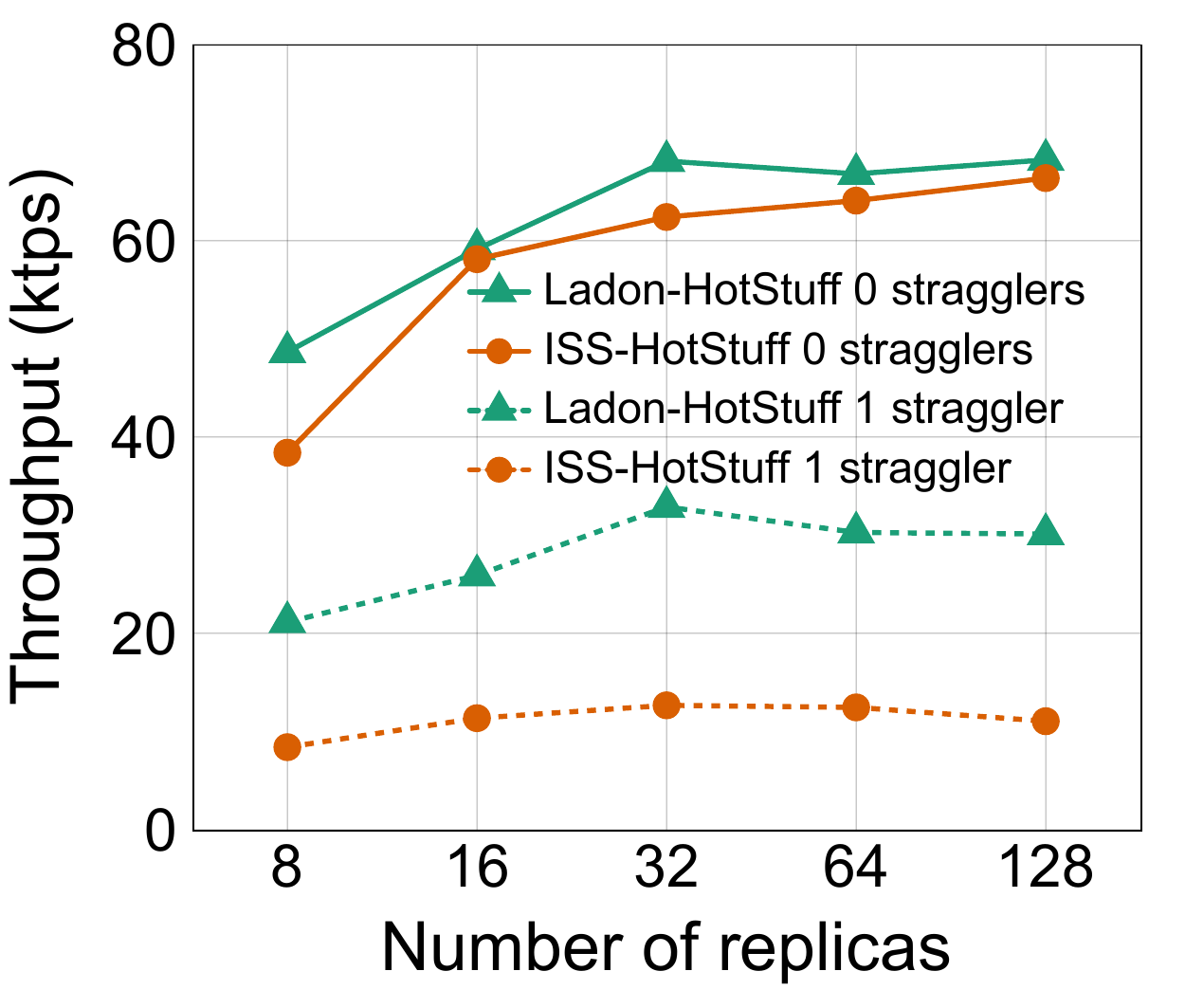}
        \caption{Throughput}
        \label{fig:throughputhotstuff}
    \end{subfigure}
    \hfill
    \begin{subfigure}[t]{0.23\textwidth}
        \centering
        \includegraphics[width=\textwidth]{ 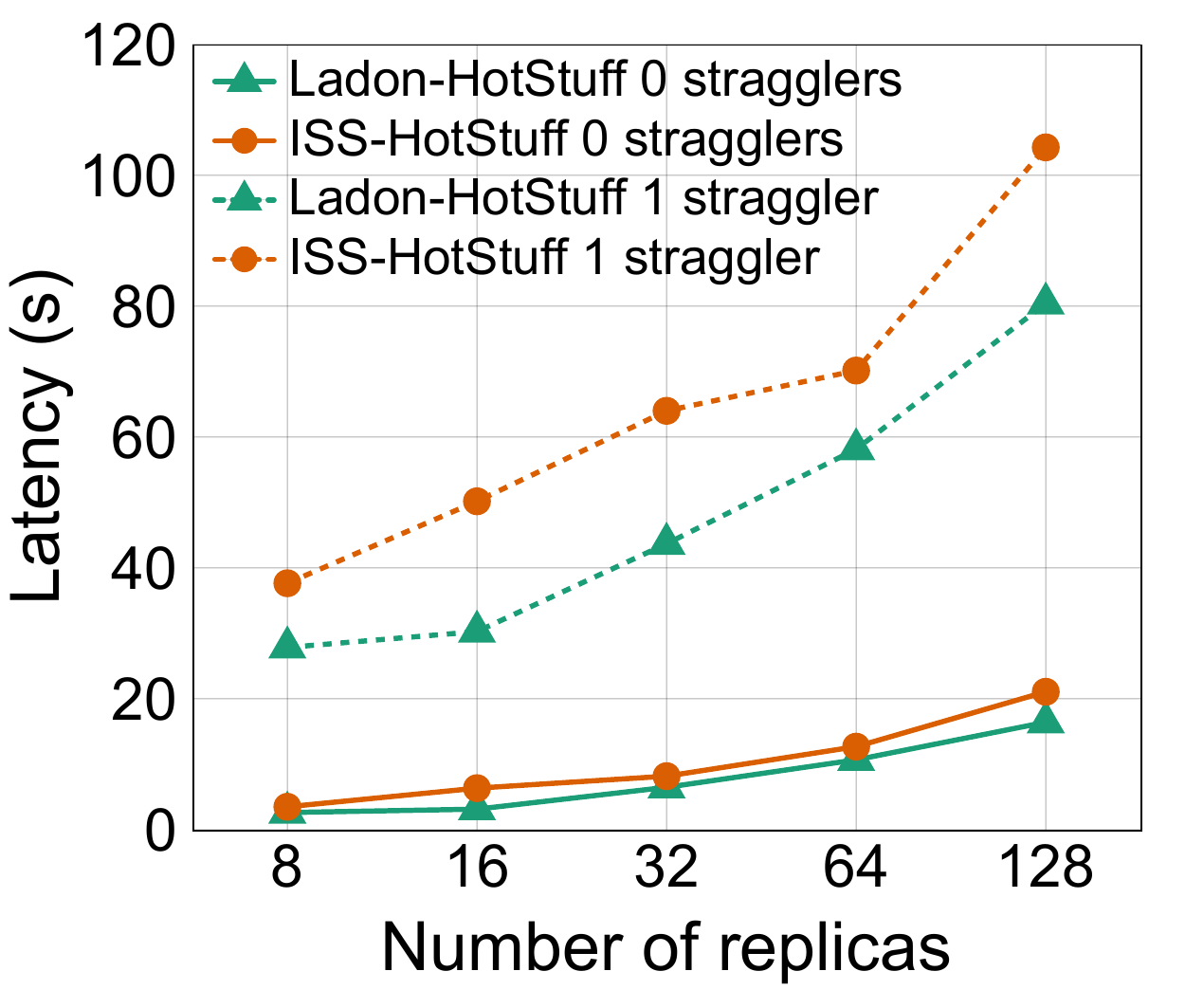}
        \caption{Latency}
        \label{fig:latencyhotstuff}
    \end{subfigure}
    \caption{\textbf{Throughput and latency of \sysname-HotStuff and ISS-HotStuff with a varying number of replicas.}}
	\label{fig:hotstuff}
\end{figure}

\figref{fig:hotstuff} shows the throughput and latency of \sysname-HotStuff and ISS-HotStuff with one honest straggler and without stragglers. The total block rate is set as 16 $blocks/s$. \figref{fig:throughputhotstuff} shows the throughput of \sysname-HotStuff without stragglers is comparable with that of ISS-HotStuff, while the throughput of  \sysname-HotStuff with one straggler is $2.7\times$ compared to that of ISS-HotStuff on 128 replicas. 
In \figref{fig:latencyhotstuff}, as the number of replicas scales up, the latency of both protocols increases. This is because we limit the total block rate of the whole system. The latency of \sysname-HotStuff without stragglers is comparable with that of ISS-HotStuff, while the latency of  \sysname-HotStuff with one straggler is 22.9\% lower than that of ISS-HotStuff on 128 replicas. 

We note that the performance of \sysname-HotStuff is more affected by stragglers compared to \sysname-PBFT. This is due to the use of chained HotStuff, where the commitment of a block is pipelined with the subsequent blocks. Consequently, blocks in a slow instance are committed much more slowly than those in PBFT, leading to reduced throughput and increased latency.

\end{document}